\documentclass[a4paper,11pt]{article}

\usepackage{graphicx}

\usepackage{fullpage}
\usepackage{libertine}
\usepackage{color}

\usepackage{multirow}

\usepackage{algorithm}
\usepackage[noend]{algpseudocode} 

\usepackage{amsmath,amsfonts,amssymb,amsthm}

\usepackage[breaklinks=true]{hyperref}
\usepackage[svgnames]{xcolor}
\usepackage[capitalise,nameinlink]{cleveref}
\hypersetup{colorlinks={true},linkcolor={DarkBlue},citecolor=[named]{DarkGreen}}

\crefname{figure}{Figure}{Figures}
\crefname{equation}{Equation}{Equations}

\newcommand{\crefthm}[1]{%
  \begingroup
  \crefname{theorem}{Thm.}{Thms.}%
  \Crefname{theorem}{Thm.}{Thms.}%
  \crefname{corollary}{Cor.}{Cors.}%
  \Crefname{corollary}{Cor.}{Cors.}%
  \cref{#1}%
  \endgroup
}

\usepackage{natbib}

\usepackage{subcaption}

\usepackage{tikz}  
\usetikzlibrary{arrows, positioning}
\usetikzlibrary{patterns,snakes}
\usetikzlibrary{decorations.shapes}
\tikzstyle{overbrace text style}=[font=\tiny, above, pos=.5, yshift=5pt]
\tikzstyle{overbrace style}=[decorate,decoration={brace,raise=5pt,amplitude=3pt}]
\usetikzlibrary{shapes.geometric}

\newtheorem{theorem}{Theorem}[section]
\newtheorem{corollary}[theorem]{Corollary}

\newtheorem{lemma}[theorem]{Lemma}

\theoremstyle{definition}
\newtheorem{definition}[theorem]{Definition}
\newtheorem*{comment*}{Comment}

\newcommand{\SC}{\text{SC}}

\newcommand{\bo}{\mathbf{o}}

\newcommand{\bx}{\mathbf{x}}
\newcommand{\favorite}{\text{top}}
\newcommand{\DIST}{\mathtt{DIST}}
\newcommand{\APP}{\mathtt{APP}}
\newcommand{\ORD}{\mathtt{ORD}}
\newcommand{\info}{\mathtt{info}}
\newcommand{\Agents}{N}
\newcommand{\NumberAgents}{n}

\newcommand{\Facilities}{F}
\newcommand{\NumberFacilities}{f}
\newcommand{\capacity}{c}
\newcommand{\calA}{\mathcal{A}}
\newcommand{\matching}{\boldsymbol{\mu}}
\newcommand{\IndividualMatching}{\mu}
\newcommand{\OptMatching}{\mathbf{o}}
\newcommand{\IndividualOpt}{o}

\newcommand{\capOf}[1]{S(#1)}
\newcommand{\agentsOf}[1]{D(#1)} 

\newcommand{\loc}{\ell} 
\newcommand{\App}{A} 

\newcommand{\subsetOfAgents}{B}
\newcommand{\ratio}{r}
\newcommand{\GenAlgorithm}{\textsc{GeneralMatch}}
\newcommand{\ADMAlg}{\textsc{VirtualMinMatch}} 
\newcommand{\AOLAlgTwo}{\textsc{TwoPriorityMatch}} 
\newcommand{\AOLAlgThreeSF}{\textsc{ThreeSameFavorite}} 
\newcommand{\AOLAlgPolarized}{\textsc{MonotoneMatch}} 

\title{\bf Metric Facility Assignment with Partial Information}

\usepackage{authblk}

\author[1]{Vasilis Gkatzelis}
\author[2]{Hasti Karimi}
\author[1]{Emma Rewinski}
\author[3]{\\Maziar Shamsipour}
\author[4,5]{Alexandros A. Voudouris}

\affil[1]{Drexel University, USA}
\affil[2]{University of British Columbia, Canada}
\affil[3]{University of Tehran, Iran}
\affil[4]{University of Essex, UK}
\affil[5]{University of Southern Denmark, Denmark}

\date{}

\begin{document}

\allowdisplaybreaks

\maketitle

\begin{abstract}
   We study an assignment problem where a set of agents and a set of facilities lie on a line metric. The goal is to compute an assignment of agents to facilities to approximately minimize the social cost (the total distance of agents from their assigned facilities) given only partial information regarding the metric. Unlike previous work which focused solely on algorithms with access to the ordinal preferences of the agents over the facilities ($\ORD$), we also consider the value of information regarding approval preferences ($\APP$), and inter-facility distances ($\DIST$). For different combinations of these three information types, we establish tight bounds on the distortion of deterministic algorithms, showing that it is possible to improve over the optimal bound of $3$ that can be achieved using only $\ORD$ information. Among other results, we show a tight bound of $1+\sqrt{2}$ for $\APP+\DIST$ which holds even for general metrics, and a tight bound of $2$ for $\ORD+\APP+\DIST$.
\end{abstract}

\section{Introduction}\label{sec:intro}
Metric facility assignment and matching problems are among the most fundamental ones in economics and computer science due to their important applications in domains such as school choice, task allocation, and resource distribution. In such problems, a set of agents must be assigned to a set of facilities located in an underlying metric space, and the standard objective is to minimize the social cost, that is, the total distance between agents and their assigned facilities. 

A central challenge is that the underlying metric information is often unavailable or prohibitively expensive to elicit. Consequently, the vast majority of studies within the social choice literature has focused on the design of algorithms that operate using only partial information about the metric, predominantly by assuming access to the {\em ordinal preferences} of the agents ($\ORD$). The performance of such ordinal algorithms is measured by the notion of {\em distortion}~\citep{distortion-survey} with respect to an objective function: the worst-case ratio between the objective value of the outcome computed by the algorithm and that of an optimal assignment with full knowledge about the metric. Recent work on the metric one-sided matching problem, where $n$ agents must be matched to $n$ facilities, has established that the best possible distortion of purely ordinal algorithms in terms of the social cost is between $\Omega(\log{n})$ (which holds even for tree metrics) and $O(n^2)$~\citep{anari2023matching}. For the special case of the line metric, the optimal distortion has been shown to be $3$~\citep{Filos-RatsikasG25}.

Although ordinal preferences are easy to elicit, they only capture relative information and ignore potentially useful cardinal structure. In practice, some additional information might be available. For example, the agents may be able to indicate the facilities they find acceptable according to a (relative) threshold, thus yielding {\em approval} information ($\APP$). In addition, the {\em inter-facility distances} ($\DIST$) may be publicly known, or easier to estimate than agent-to-facility distances; this is true in scenarios where the facilities correspond to physical locations. Using $\ORD$ and $\DIST$ information, \citet{anshelevich2021given} showed a best possible distortion bound of $3$ that holds for a quite extensive class of problems (that includes one-sided matching) and for arbitrary metrics. Combinations of $\ORD$, $\DIST$, and $\APP$ were recently considered by \citet{Anshelevich2025approval} for metric single-winner voting, where agents have preferences over a set of candidates, and the goal is to choose one of them as the winner.

In this work, we study a metric facility assignment problem that generalizes one-sided matching and investigate the power of combining these three information types: $\ORD$, $\APP$, and $\DIST$. We focus on deterministic algorithms and primarily the line metric. Our main research question is: 
\begin{center}
  \em  How does the distortion of matching algorithms change\\ under different combinations of these information types?
\end{center}
Our results reveal a nearly complete characterization of the optimal distortion achievable under these informational assumptions, showing that, in some cases, substantial improvements can be achieved relative to the best possible bounds using solely $\ORD$ (or in combination to $\DIST$) information. 

\subsection{Our Contribution}
To be more specific, we study a metric facility assignment problem in which each facility has a fixed capacity such that the total capacity across all facilities equals the number $n$ of agents. Given partial information about the underlying (line) metric, the goal is to compute an assignment of agents to facilities that satisfies their capacity constraints, and approximately minimizes the total social cost (the sum of the distances between agents and their assigned facilities). We consider algorithms that have access to combinations of the three types of partial information mentioned above. The first one is ordinal information ($\ORD$), where each agent reports a ranking of the facilities from the closest to the farthest one. 
The second is approval information ($\APP$), where, for a parameter $\alpha \geq 1$ that can be fine-tuned by the algorithm designer, each agent reports an approval set consisting of every facility whose distance is at most $\alpha$ times the distance to her closest facility. 
The third is inter-facility distance information ($\DIST$), namely the pairwise distances between facilities. 

\renewcommand{\arraystretch}{1.4}
\begin{table}[t]
    \centering
    \begin{tabular}{l|l|l}
                                & $f=2$ & $f \geq 3$  \\\hline
    $(\DIST +) \ORD$            & \multicolumn{2}{|c}{$3$ \citep{Filos-RatsikasG25}}         \\\hline
    $\DIST$, $\APP$             & \multicolumn{2}{|c}{$+\infty$ (\crefthm{thm:distances-or-approvals:unbounded})}  \\\hline
    $\DIST + \APP$              & $2$ (\crefthm{thm:distances+approvals:two-locations:lower,thm:distances+approvals:two-locations:upper}) & $1+\sqrt{2}$$\,^\star$ (\crefthm{thm:distances+approvals:lower,cor:distances+approvals:upper-optimized}) \\\hline
    $\DIST + \APP + \ORD$     &  $2\sqrt{2}-1$ (\crefthm{cor:distances+approvals+ordinal:two-locations:optimized-upper,thm:distances+approvals+ordinal:two-locations:lower}) &  $2$ (\crefthm{thm:distances+approvals+ordinal:three-locations:lower,thm:distances+approvals+ordinal:general-upper})   \\\hline
    \end{tabular}
    \caption{An overview of our tight distortion bounds for the various combinations of information types we consider (rows) in terms of the number $f$ of facilities (columns). All results hold for the line metric with the notable exception of the $1+\sqrt{2}$ bound for $\DIST + \APP$ marked with $^\star$, which holds for arbitrary metrics. In the first row, shown as $(\DIST +) \ORD$, the upper bound follows by an algorithm that uses $\ORD$ information only, whereas the lower bound holds for all algorithms that use $\DIST+\ORD$ information. In the last row, the upper bound of $2\sqrt{2}-1$ for $f=2$ facilities does not require $\DIST$ and is tight even over algorithms that use $\DIST$. The upper bound of $2$ for $f\geq 3$ facilities requires the ordering of the facilities on the line rather than full $\DIST$ information.}
    \label{tab:overview}
\end{table}

Our main contribution is a characterization of the optimal distortion achievable by deterministic algorithms for combinations of these information types, with a sharp dependence on the number of facilities. An overview of our results is given in \cref{tab:overview}. Since a tight distortion bound of $3$ is achievable by algorithms with access to only $\ORD$ information~\citep{Filos-RatsikasG25}, we first consider either $\APP$ or $\DIST$ in isolation, and show a strong impossibility: the distortion is unbounded. We then consider algorithms with access to $\APP+\DIST$ (but not $\ORD$). By carefully adapting to our setting and optimizing (over $\alpha$) the algorithm of \citet{Anshelevich2025approval} for single-winner voting under the same information assumptions, we show an upper bound of $1+\sqrt{2}\approx 2.41$ that holds for arbitrary metrics, and is best possible over all algorithms that use such information, even for the line metric. For instances with two facilities, the same algorithm (but with a different value for the parameter $\alpha$) yields a distortion of $2$, which is optimal within the class of algorithms with the same type of information. 

We next consider the richest informational model $\APP+\DIST+\ORD$ for line metrics. For instances with two facilities, we establish a tight bound of $2\sqrt{2}-1\approx 1.83$. The upper bound follows by optimizing an algorithm that prioritizes the over-demanded facility, and assigns to it the agents that prefer it with priority given to those with smaller approval sets; interestingly, this algorithm does not require knowing the exact distance between the two facilities, yet it is optimal over all algorithms that might use this extra information. Our most technically demanding result is for instances of three or more facilities, where we establish a tight bound of $2$. The upper bound follows by an algorithm which, by appropriately using the ordering of the facilities on the line (induced by the $\DIST$ information), recursively partitions any given instance into either demand-monotone sub-instances that can be solved using an extension of our two-facility algorithm, or three-facility sub-instances where all agents rank the same facility first. We finally observe that using $\DIST$ in combination to $\APP+\ORD$ is necessary to achieve a bound better than $3$ for at least four facilities; otherwise, using only $\APP+\ORD$, there is a lower bound of $3$ showing that this informational model is not stronger than solely using $\ORD$.

\subsection{Related Work}
Our work belongs to the broader literature on distortion in social choice problems, where, as already discussed, the goal is to evaluate the loss incurred by algorithms that make decisions using only partial information, typically of ordinal nature. 
Starting with the work of \citet{procaccia2006distortion}, the distortion framework has been studied extensively for several prominent problems, including 
utilitarian voting settings~\citep{boutilier2015optimal,caragiannis2011voting,caragiannis2017subset,ebadian2024optimized}, 
voting with metric preferences~\citep{anshelevich2018approximating,gkatzelis2020resolving,kempe2022veto,caragiannis2022multiwinner,charikar2022randomized,charikar24breaking}, 
participatory budgeting~\citep{benade2021preference,Bedaywi2025public-participatory}, 
and clustering~\citep{burkhardt2024low}.

The line of work most closely related to ours is that on distortion of metric minimum-cost one-sided matching. The first paper to consider this direction was by~\citet{caragiannis2024augmentation} who focused on strategyproof mechanisms; among other results, they showed that Serial Dictatorship (which is an ordinal algorithm) achieves an exponential distortion in the number $n$ of agents, while its randomized counterpart, Randomized Serial Dictatorship, achieves an expected distortion between $n^{0.29}$ and $n$. Later on, \citet{anshelevich2021given} proved a tight distortion bound of $3$ for ordinal algorithms that have additional access to inter-facility distances. For only ordinal information, \citet{anari2023matching} designed a clustering-based deterministic algorithm with distortion $O(n^2)$ and proved a lower bound of $\Omega(\log{n})$ on the distortion of any ordinal algorithm (including randomized ones) that holds even for tree metrics. More recently, \citet{Filos-RatsikasG25} focused on the line metric and showed that a best possible constant distortion of $3$ can be achieved with only ordinal information for the more general $k$-centrum cost (the sum of the $k$ largest distances) which includes the social cost for $k=n$. \citet{hastings2025fairmetricdistortionmatching} adapted the algorithm of \citep{anari2023matching} and showed distortion bounds for fairness notions (including the $k$-centrum cost) for arbitrary metrics. 

Matching problems with metric utilities rather than costs, where the goal is to match agents to facilities to maximize the total utility (possibly under restrictions such as truthfulness), have also been considered~\citep{Anshelevich2019bipartitematching,Anshelevich2016truthful,anshelevich2016blind}. There is also a large related literature on matching problems without an underlying metric structure, in which agents have (normalized) valuation functions. This includes work on one-sided, two-sided, and several other matching problems using algorithms that have access to ordinal, limited cardinal (obtained via e.g. value queries), randomized, or approval-based preference information~\citep{filos2014RSD,amanatidis2022matching,amanatidis2024dice,ma2021matching,latifian2024approval,ebadian2025bit,Filos-Ratsikas2026stable}. Moving beyond worst-case analysis, recent work has also analyzed the distortion of matching algorithms in stochastically generated instances~\citep{DENG2022104920,GaoZhang19,caragiannis2026distortionpriorindependent}

\section{Preliminaries}\label{sec:prelim}
We consider a {\em facility assignment} problem with a set $\Agents$ of $\NumberAgents \geq 2$ {\em agents} and a set $\Facilities$ of $\NumberFacilities \geq 2$ {\em facilities}. Each facility $x \in \Facilities$ has a {\em capacity} $\capacity_x \in \mathbb{N}$ such that $\sum_{x \in \Facilities} \capacity_x = \NumberAgents$. Agents and facilities are represented by points in an underlying {\em unknown} metric space; that is, the distances between them satisfy the triangle inequality: $d(x,y) \leq d(x,z) + d(z,y)$ for any $x,y,z \in \Agents \cup \Facilities$. While some of our results hold for arbitrary metrics, we primarily focus on the case where the metric is a {\em line}; then, each facility has a location $\loc_x\in \mathbb{R}$. 
An {\em algorithm} takes as input some information $\info(d)$ about the metric space $d$ and computes a feasible {\em assignment}  
$\matching = (\IndividualMatching_i)_i$ of agents to facilities such that, for any facility $x \in \Facilities$, the total number of agents assigned to it equals its capacity, i.e., $|\{i \in \Agents: \IndividualMatching_i = x\}| = c_x$. Observe that each capacity unit of a facility can be viewed as a different resource that is allocated to a unique agent. Hence, an assignment is equivalent to a one-to-one matching between agents and resources, and our facility assignment problem generalizes the {\em one-sided matching} problem. We will use the terms {\em assignment} and {\em matching} interchangeably throughout the remaining of the paper. 

We consider algorithms with access to different combinations of the following types of information: 
\begin{itemize}
    \item $\DIST$: The distances between the facilities in the metric space.
    \item $\APP$: Each agent $i$ reports an {\em $\alpha$-threshold approval set} $A_i$, which consists of any facility $x$ such that $d(i,x) \leq \alpha \cdot d(i,\favorite_i)$, where $\favorite_i$ is the facility closest to agent $i$. We will interchangeably refer to $A_i$ as the approval set or the {\em approval preference} of $i$. 
    \item $\ORD$: Each agent $i$ reports an {\em ordering} $\succ_i$ over the facilities, ranked from the closest to the farthest one. We will sometimes refer to $\succ_i$ as the {\em ordinal preference} of $i$. 
\end{itemize}
The goal is to compute a matching $\matching$ so as to minimize the {\em social cost}, defined as the {\em total} distance of the agents from their matched facilities: 
\begin{align*}
    \SC(\matching) = \sum_{i\in \Agents} d(i,\IndividualMatching_i).
\end{align*} 
Due to the limited information about the metric space that our algorithms have access to, it is impossible to minimize the social cost exactly, and some loss of efficiency is inevitable. The {\em distortion} of an algorithm $\calA$ that uses information $\info(d)$ is the worst-case (over all metric spaces) ratio between the social cost of the computed matching $\calA(\info(d))$ and the minimum social cost over all possible matchings:
\begin{align*}
    \sup_{d} \frac{\SC(\calA(\info(d)))}{\min_{\matching} \SC(\matching)}.
\end{align*}
By definition, the distortion is at least $1$, and the goal is to design algorithms with as small distortion as possible. 

When the metric space is a line, \citet{Filos-RatsikasG25} focused on an optimal matching (i.e. a matching that minimizes the social cost) with a very particular and useful structure: the agents are matched to facilities from left to right according to their true orderings on the line. We state this in the following lemma as it will come in handy in several of our results. 

\begin{lemma}[\citep{Filos-RatsikasG25}]\label{lem:line:optimal}
Given the true ordering of all agents and facilities on the line, there is an optimal matching that greedily matches the leftmost agent to the leftmost facility.
\end{lemma}

\section{Combining Distance and Approval Information}\label{sec:distances+approvals}

In this section we focus on algorithms that can use information about the distances between the facilities ($\DIST$) and approval information about the preferences of the agents over the facilities ($\APP$). We start by showing that using only one of these types of information is not sufficient to achieve bounded distortion. 

\begin{theorem} \label{thm:distances-or-approvals:unbounded}
    The distortion of any algorithm that uses only $\DIST$ information or only $\APP$ information is unbounded. 
\end{theorem}

\begin{proof}
First consider an arbitrary algorithm that has access to DIST information but has no information about the preferences of the agents over the facilities. Consider an instance with two unit-capacity facilities $x$ and $y$, and two agents $1$ and $2$. Suppose the algorithm matches agent $1$ to $x$ and agent $2$ to $y$. Then, it might be the case that agent $1$ is co-located with facility $y$, and agent $2$ is co-located with facility $x$. So, the optimal social cost is $0$, but the social cost of the matching computed by the algorithm is positive, thus leading to unbounded distortion. 

We next focus on algorithms that use only $\APP$ information. Let $\varepsilon > 0$ be an infinitesimal.
Consider an instance with $4$ agents and $4$ unit-capacity facilities $\{x,y,w,z\}$. The first two agents approve only of $x$ and the remaining two agents approve only of $y$. Let $\matching$ be the matching computed by the algorithm and, without loss of generality, suppose that agent $1$ is matched to $w$. We place the agents and the facilities on the line as follows:
\begin{itemize}
    \item Agent $1$, agent $2$, and facility $x$ are at $0$;
    \item facility $z$ is at $\varepsilon$;
    \item facility $w$ is at $1-\varepsilon$;
    \item Agent $3$, agent $4$, and facility $y$ are at $1$.
\end{itemize}
Observe that all agents are co-located with the facility they approve, and thus this line metric is consistent to the revealed approval information. The social cost of the matching computed by the algorithm is at least $d(1,w) = 1-\varepsilon$. In contrast, the optimal social cost is $2\varepsilon$ and is achieved by matching $1$ to $x$, $2$ to $z$, $3$ to $w$, and $4$ to $y$. So, the distortion becomes unbounded when $\varepsilon$ tends to $0$.  
\end{proof}

Despite the strong impossibility when using just $\DIST$ or just $\APP$ information, we next show that we can achieve small constant distortion by combining them. In particular, building on the ideas of \citet{Anshelevich2025approval} for single-winner voting under the same informational assumptions, we prove an upper bound of $1+\sqrt{2}$ on the distortion of matching algorithms that use $\DIST$ and $\APP$ information. In particular, given this type of information, our algorithm creates a modified full-information instance in which each agent $i$ is at a {\em virtual distance} $\tilde{d}(i,x) = \min_{y \in A_i} d(y,x)$ from each facility $x$. It then solves the modified minimum-weight matching instance optimally using, e.g., the Hungarian algorithm, and finally outputs the computed matching as a solution for the original instance. See \cref{alg:virtual-min-matching}. 

\begin{algorithm}
\caption{$\alpha$-\ADMAlg}
\label{alg:virtual-min-matching}
\begin{algorithmic}[1]
\Require $\DIST$ and $\APP$ information for some threshold $\alpha$.
\Ensure Matching $\matching$.
\For{each agent $i \in \Agents$ and facility $x \in \Facilities$}
   \State $y_i(x) \gets \arg\min_{y \in A_i} d(y,x)$.
   \State $\tilde{d}(i,x) \gets d(y_i(x),x)$.
\EndFor
\State $\matching \gets \arg\min_{\bx = (x_i)_{i \in \Agents}} \sum_{i \in \Agents} \tilde{d}(i,x_i)$.
\State \Return $\matching$.
\end{algorithmic} 
\end{algorithm}

\begin{theorem} \label{thm:distances+approvals:upper}
    The distortion of {\sc $\alpha$-\ADMAlg} is at most $\max\left\{\alpha, 2 + 1/\alpha\right\}$.
\end{theorem}

\begin{proof}
Denote by $\matching=(\IndividualMatching_i)_i$ the matching that is computed by the algorithm, 
and by $\OptMatching=(\IndividualOpt_i)_i$ an optimal matching. 
Since $\matching$ minimizes the social cost with respect to the virtual distances $\tilde{d}$, we have
\begin{align*}
    \sum_i \tilde{d}(i,\IndividualMatching_i) \leq  \sum_i \tilde{d}(i,\IndividualOpt_i) \Leftrightarrow \sum_i d(y_i(\IndividualMatching_i),\IndividualMatching_i) \leq \sum_i d(y_i(\IndividualOpt_i),\IndividualOpt_i).
\end{align*}
Using the  triangle inequality, we can bound the social cost of the algorithm as follows:
\begin{align}
\SC(\matching) = \sum_i d(i,\IndividualMatching_i)  &\leq \sum_i \bigg( d(i,y_i(\IndividualMatching_i)) + d(y_i(\IndividualMatching_i),\IndividualMatching_i) \bigg) \nonumber \\
&\leq \sum_i \bigg( d(i,y_i(\IndividualMatching_i)) +  d(y_i(\IndividualOpt_i),\IndividualOpt_i) \bigg). \label{eq:distances+approvals:upper:main-inequality}
\end{align}
We now make the following observations for the agents depending on whether their optimal facility is in their approval set or not.
\begin{itemize}
    \item Agent $i$ such that $\IndividualOpt_i \in A_i$. Then, $y_i(\IndividualOpt_i) = \IndividualOpt_i$, and thus $d(y_i(\IndividualOpt_i),\IndividualOpt_i) = 0$. Hence, 
    \begin{align*}
        d(i,y_i(\IndividualMatching_i)) +  d(y_i(\IndividualOpt_i),\IndividualOpt_i) 
        &= d(i,y_i(\IndividualMatching_i)) \\
        &\leq \alpha \cdot d(i,\favorite_i)
        \leq \alpha \cdot d(i,\IndividualOpt_i).
    \end{align*}

    \item Agent $i$ such that $\IndividualOpt_i \not\in A_i$. By definition, $y_i(\IndividualOpt_i)$ is the closest facility in $A_i$ to $\IndividualOpt_i$, and thus $d(y_i(\IndividualOpt_i),\IndividualOpt_i) \leq d(\favorite_i,\IndividualOpt_i)$. In addition, since $\IndividualOpt_i \not\in A_i$, $d(i,\IndividualOpt_i) > \alpha \cdot d(i,\favorite_i)$, and thus $d(i,\favorite_i) < \frac{1}{\alpha} \cdot d(i,\IndividualOpt_i)$. Using these, and the triangle inequality, we have
    \begin{align*}
         d(i,y_i(\IndividualMatching_i)) +  d(y_i(\IndividualOpt_i),\IndividualOpt_i) 
         &\leq \alpha \cdot d(i,\favorite_i) + d(\favorite_i,\IndividualOpt_i) \\
         &\leq (1+\alpha) \cdot d(i,\favorite_i) + d(i,\IndividualOpt_i) \\
         &\leq \left(\frac{1}{\alpha} + 1\right) d(i,\IndividualOpt_i) + d(i,\IndividualOpt_i) \\
         &= \left(2 + \frac{1}{\alpha}\right) \cdot d(i,\IndividualOpt_i).
    \end{align*}
\end{itemize}
Consequently, we have that $\SC(\matching) \leq \max\left\{\alpha, 2 + 1/\alpha\right\} \cdot \SC(\OptMatching)$, and the proof is complete. 
\end{proof}

By optimizing over $\alpha$, \cref{thm:distances+approvals:upper} yields the following.

\begin{corollary}\label{cor:distances+approvals:upper-optimized}
    The distortion of {\sc $(1+\sqrt{2})$-\ADMAlg} is at most $1+\sqrt{2} \approx 2.414$.
\end{corollary}

We next show that the $1+\sqrt{2}$ distortion bound is best possible over all algorithms that use $\DIST+\APP$ information. 


\begin{theorem} \label{thm:distances+approvals:lower}
The distortion of any matching algorithm that uses $\DIST+\APP$ information is at least $1+\sqrt{2}$, even when there are $\NumberFacilities=3$ facilities. 
\end{theorem}

\begin{proof}
Let $\varepsilon > 0$ be an infinitesimal and consider the following instance: 
\begin{itemize}
    \item There is a facility $L$ with capacity $c_L = (n-1)/2$ at $-\alpha-\varepsilon$;
    \item There is a unit-capacity facility $x$ at $1$;
    \item There is a facility $R$ with capacity $c_R = (n-1)/2$ at $\alpha$.
\end{itemize}
All $n$ agents approve $x$ and $R$, that is, $A_i=\{x, R\}$ for any $i \in [n]$. 

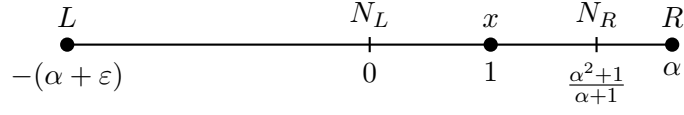
\begin{figure}[t]
    \centering
    \begin{tikzpicture}[scale=1]
        \draw[thick] (-4,0) -- (4,0);

        \draw[thick,fill=black] (-4,0) circle (0.08);
        \node[below] at (-4, -0.1) {$-(\alpha+\varepsilon)$};
        \node[above] at (-4, 0.1) {$L$};

        \draw[thick] (0,0.1) -- (0, -0.1) node[below] {$0$};
        \node[above] at (0, 0.1) {$N_L$};

        \draw[thick,fill=black] (1.6,0) circle (0.08);
        \node[below] at (1.6, -0.1) {$1$};
        \node[above] at (1.6, 0.1) {$x$};

        \draw[thick] (3,0.1) -- (3, -0.1) node[below] {$\frac{\alpha^2+1}{\alpha+1}$};
        \node[above] at (3, 0.1) {$N_R$};

        \draw[thick,fill=black] (4,0) circle (0.08);
        \node[below] at (4, -0.1) {$\alpha$};
        \node[above] at (4, 0.1) {$R$};
    \end{tikzpicture}
\caption{The metric space considered in the proof of \cref{thm:distances+approvals:lower}.}
    \label{fig:distances+approvals:lower}
\end{figure}

Consider any matching $\matching$ that might be computed by an arbitrary algorithm. We place the agents on the line as follows (see \cref{fig:distances+approvals:lower}).
\begin{itemize}
    \item The agents that are matched to $R$ are located at $0$; let $N_L$ be the set of these agents. 
    Observe that for any $i\in N_L$, we have that
    \begin{align*}
    & d(i,L) = \alpha+\varepsilon \\
    & d(i,x) = 1; \\
    & d(i,R) = \alpha.
    \end{align*}
    Hence, this position of $i$ is consistent to her approval set. 
    
    \item The agents that are matched to $x$ and $L$ are located at $\frac{\alpha^2+1}{\alpha+1}$; let $N_R$ be the set of these agents. For any $i \in N_R$, we have that
    \begin{align*}
    & d(i,L) = \frac{2\alpha^2+\alpha+1}{\alpha+1}+\delta, \text{, where $\delta$ is some infinitesimal}; \\
    & d(i,x)= \alpha \cdot \frac{\alpha-1}{\alpha+1}; \\
    & d(i,R)= \frac{\alpha-1}{\alpha+1}.
    \end{align*}
    Hence, this position of $i$ is consistent to her approval set. 
\end{itemize}
The social cost of the computed matching is 
\begin{align*}
    \SC(\matching) 
    =\frac{n-1}{2} \cdot \alpha + \alpha \cdot \frac{\alpha-1}{\alpha+1} + \frac{n-1}{2} \cdot \left( \frac{2\alpha^2+\alpha+1}{\alpha+1}+\delta \right).
\end{align*}
On the other hand, the optimal matching $\OptMatching$ is to match the $(n-1)/2$ the agents in $N_L$ to $L$, one of the agents in $N_R$ to $x$, and the remaining $(n-1)/2$ agents in $N_R$ to $R$. The optimal social cost is 
\begin{align*}
    \SC(\OptMatching) = \frac{n-1}{2} \cdot (\alpha+\varepsilon) + \alpha \cdot \frac{\alpha-1}{\alpha+1} + \frac{n-1}{2} \cdot \frac{\alpha-1}{\alpha+1}.
\end{align*}
As $\varepsilon$ tends to $0$ and $n$ tends to infinity, the distortion tends to 
\begin{align*}
    \frac{\alpha + \frac{2\alpha^2+\alpha+1}{\alpha+1}}{ \alpha + \frac{\alpha-1}{\alpha+1}} = \frac{3\alpha^2 + 2\alpha+1}{\alpha^2 + 2\alpha-1}.
\end{align*}
This function attains its minimum value of $1+\sqrt{2}$ for $\alpha=1+\sqrt{2}$, and the proof is complete. 
\end{proof}

At this point we can make two crucial observations. First, the analysis in the proof of \cref{thm:distances+approvals:upper} does not require the agents and the facilities to be located on a line. Therefore, that proof works even for arbitrary metric spaces, and the algorithm achieves a distortion bound of $1+\sqrt{2}$ more generally. Second, this bound is tight even when the metric is a line, and the example used in the proof of \cref{thm:distances+approvals:lower} to show this requires three facilities. This leaves open the possibility that an improved distortion bound can be achieved for the case of two facilities. We show that this is indeed the case.

\begin{theorem} \label{thm:distances+approvals:two-locations:upper}
For instances with $\NumberFacilities=2$ facilities, {\sc $2$-\ADMAlg}   achieves a distortion of at most $2$.
\end{theorem}

\begin{proof}
Let $\alpha = 2$.
Following the proof of \cref{thm:distances+approvals:upper}, we want to bound \cref{eq:distances+approvals:upper:main-inequality}:
\begin{align*}
\SC(\matching) \leq \sum_i \bigg( d(i,y_i(\IndividualMatching_i)) +  d(y_i(\IndividualOpt_i),\IndividualOpt_i) \bigg). 
\end{align*}
We again consider two cases depending on whether the optimal facility of an agent is in her approval set or not. 
\begin{itemize}
    \item Agent $i$ such that $\IndividualOpt_i \in A_i$. 
    Then, $y_i(\IndividualOpt_i) = \IndividualOpt_i$, and thus $d(y_i(\IndividualOpt_i),\IndividualOpt_i) = 0$. Hence, 
    \begin{align*}
        d(i,y_i(\IndividualMatching_i)) +  d(y_i(\IndividualOpt_i),\IndividualOpt_i) 
        &= d(i,y_i(\IndividualMatching_i)) \\
        &\leq \alpha \cdot d(i,\favorite_i)
        \leq \alpha \cdot d(i,\IndividualOpt_i) 
        = 2 \cdot d(i,\IndividualOpt_i).
    \end{align*}

    \item Agent $i$ such that $\IndividualOpt_i \not\in A_i$. 
    Since $|A_i| \leq 2$, then it must be the case that $|A_i| = 1$, and thus $y_i(\IndividualMatching_i)) = y_i(\IndividualOpt_i)) = \favorite_i$. Since $\IndividualOpt_i \not\in A_i$, we also have that $d(i,\IndividualOpt_i) > \alpha \cdot d(i,\favorite_i)$, and thus $d(i,\favorite_i) < \frac{1}{\alpha} \cdot d(i,\IndividualOpt_i)$. Using these, and the triangle inequality, we have
    \begin{align*}
         d(i,y_i(\IndividualMatching_i)) +  d(y_i(\IndividualOpt_i),\IndividualOpt_i) 
         &=  d(i,\favorite_i) + d(\favorite_i,\IndividualOpt_i) \\
         &\leq 2 \cdot d(i,\favorite_i) + d(i,\IndividualOpt_i) \\
         &\leq \frac{2}{\alpha}\cdot d(i,\IndividualOpt_i) + d(i,\IndividualOpt_i) \\
         &= 2\cdot d(i,\IndividualOpt_i). 
    \end{align*}
\end{itemize}
Consequently, we have that $\SC(\matching) \leq 2 \cdot \SC(\OptMatching)$, and the proof is complete. 
\end{proof}

\begin{theorem} \label{thm:distances+approvals:two-locations:lower}
The distortion of any algorithm that uses $\DIST+\APP$ information is at least $2$, even when there are $\NumberFacilities=2$ facilities. 
\end{theorem}

\begin{proof}
For any $\alpha \geq 2$, consider an instance with a unit-capacity facility $x$ located at $0$, a unit-capacity facility $y$ located at $1$, and two agents that approve both facilities. 
Without loss of generality, suppose that the algorithm outputs the matching $\matching = (x,y)$, i.e., agent $1$ is matched to $x$ and agent $2$ is matched to $y$. We then place the agents on the line as follows (see \cref{fig:distances+approvals:two-locations:lower:a}):
\begin{itemize}
    \item Agent $1$ is located at $\alpha/(\alpha-1)$. Since $d(1,y) = \alpha/(\alpha-1) - 1 = 1/(\alpha-1)$ and $d(1,x) = \alpha/(\alpha-1)$, this location is consistent with $A_1 = \{x,y\}$.

    \item Agent $2$ is located at $1/(\alpha+1)$. Since $d(2,x) = 1/(\alpha+1)$  and $d(2,y) = 1 - 1/(\alpha+1) = \alpha/(\alpha+1)$, this location is consistent with $A_2 = \{x,y\}$.
\end{itemize}
The social cost of the algorithm is 
$$\SC(\matching) = d(1,x) + d(2,y) = \alpha \cdot \bigg(\frac{1}{\alpha-1} + \frac{1}{\alpha+1}\bigg).$$ 
On the other hand, the optimal matching is $\OptMatching = (y,x)$, i.e., agent $1$ is matched to $y$ and agent $2$ is matched to $x$. 
The optimal social cost is 
$$\SC(\OptMatching) = d(1,y) + d(2,x) = \frac{1}{\alpha-1} + \frac{1}{\alpha+1},$$ 
leading to a distortion of $\alpha \geq 2$.

\begin{figure}[t]
\centering
\begin{subfigure}[t]{0.45\linewidth}
\centering
\begin{tikzpicture}[scale=1]
  \draw[thick] (-1.5,0) -- (3,0);

  \draw[thick,fill=black] (-1.5,0) circle (0.08);
  \node[below] at (-1.5,-0.2) {$0$};
  \node[above] at (-1.5,0.2) {$x$};

  \draw[thick,fill=black] (3,0) circle (0.08);
  \node[below] at (3,-0.2) {$1$};
  \node[above] at (3,0.2) {$y$};

  \draw[thick] (1.5,0.1) -- (1.5,-0.1);
  \node[below] at (1.5,-0.2) {$\frac{\alpha}{\alpha+1}$};
  \node[above] at (1.5,0.25) {$1$};

  \draw[thick] (0,0.1) -- (0,-0.1);
  \node[below] at (0,-0.2) {$\frac{1}{\alpha-1}$};
  \node[above] at (0,0.25) {$2$};
\end{tikzpicture}
\caption{}
\label{fig:distances+approvals:two-locations:lower:a}
\end{subfigure}
\begin{subfigure}[t]{0.45\linewidth}
\centering
\begin{tikzpicture}[scale=1]
  \draw[thick] (-1.5,0) -- (3,0);

  \draw[thick,fill=black] (-1.5,0) circle (0.08);
  \node[below] at (-1.5,-0.2) {$0$};
  \node[above] at (-1.5,0.2) {$x,2$};

  \draw[thick,fill=black] (3,0) circle (0.08);
  \node[below] at (3,-0.2) {$1$};
  \node[above] at (3,0.2) {$y$};

  \draw[thick] (0,0.1) -- (0,-0.1);
  \node[below] at (0,-0.2) {$\frac{1}{\alpha-1}-\varepsilon$};
  \node[above] at (0,0.25) {$1$};
\end{tikzpicture}
\caption{}
\label{fig:distances+approvals:two-locations:lower:b}
\end{subfigure}
\caption{The two instances used in the proof of \cref{thm:distances+approvals:two-locations:lower}.}
\label{fig:distances+approvals:two-locations:lower}
\end{figure}
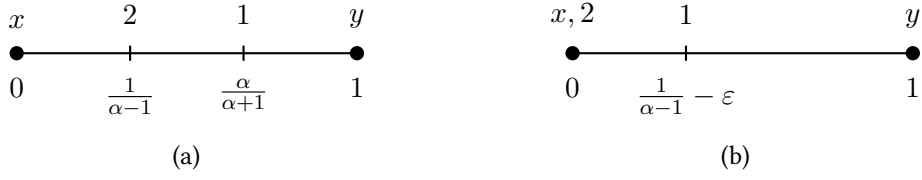

For any $\alpha < 2$, consider an instance with a unit-capacity facility $x$ located at $0$, a unit-capacity facility $y$ located at $1$, and two agents that approve only $x$.
Without loss of generality, suppose that the algorithm outputs the matching $\matching = (x,y)$, i.e., agent $1$ is matched to $x$ and agent $2$ is matched to $y$. We place the agents on the line as follows (see \cref{fig:distances+approvals:two-locations:lower:b}): 
\begin{itemize}
\item Agent $1$ is located at $1/(\alpha+1)-\varepsilon$, for some infinitesimal $\varepsilon > 0$. Note that $d(1,x) = 1/(\alpha+1)-\varepsilon$ and $d(1,y) = \alpha/(\alpha+1)+\varepsilon$; hence, this location is consistent with $A_1 = \{x\}$.
\item Agent $2$ is located at $0$. Note that $d(2,x) = 0$  and $d(2,y) = 1$; hence, this location is consistent with $A_2 = \{x\}$.
\end{itemize}
The social cost of the algorithm is 
$$\SC(\matching) = d(1,x) + d(2,y) = \frac{1}{\alpha+1}-\varepsilon + 1  = \frac{\alpha+2}{\alpha+1} - \varepsilon.$$ 
On the other hand, the optimal matching $\OptMatching$ is to match agent $1$ to $y$ and agent $2$ to $x$ for a social cost of 
$$\SC(\OptMatching) = d(1,y) + d(2,x) = \frac{\alpha}{\alpha+1}+\varepsilon.$$ 
As $\varepsilon$ tends to $0$, since $\alpha < 2$, the distortion is $1 + 2/\alpha > 2$.
\end{proof}

\section{Two Facilities with Distance, Approval, and Ordinal Information}
\label{sec:distances+approvals+ordinal:two}
We now switch to the richer space of using all three information types ($\DIST$, $\APP$ and $\ORD$). 
We first consider instances with two facilities, for which we establish a tight bound of $2\sqrt{2}-1\approx 1.83$. In fact our result is stronger: The upper bound follows by an algorithm that uses only $\APP+\ORD$ information, whereas the lower bound is tight over all algorithms that use $\DIST+\APP+\ORD$ information. 

We start with the upper bound. 
Let $x$ and $y$ be the two facilities with capacities $c_x$ and $c_y$. Our algorithm partitions the agents according to their ordinal and approval preferences as follows; note that information about the distance between the two facilities is not required. For any facility $z \in \{x,y\}$, let $n_z$ be the number of agents that prefer it over the other facility, and, for $k \in [2]$, let 
\begin{align*}
    \Agents_z^k = \{i \in N : \favorite_i = z \text{ and } |A_i|=k  \}.
\end{align*}
Clearly, we have $n_z = |\Agents_z^1|+|\Agents_z^2|$.
The algorithm starts by matching the agents to the facilities in a specific order which depends on the facility preferred by most agents and the partition of the agents into the aforementioned sets. Let $w$ be the {\em over-demanded} facility out of the two, that is, $w \in \arg\max_{z\in\{x,y\}}\{n_z - c_z\}$; note that, since $c_x + c_y = n$, one of the two facilities must be (weakly) over-demanded. 
The agents are assigned to facility $w$ with priority given to the agents of $\Agents_{w}^1$, followed by the agents of $\Agents_{w}^2$. When the capacity of $w$ is depleted (which, by definition, happens before the agents in $\Agents_{w}^1 \cup \Agents_{w}^2$ are depleted), the remaining agents are assigned to the other facility.
See \cref{alg:DAO-2locations}.

\begin{algorithm}
\caption{{\sc $\alpha$-\AOLAlgTwo}}
\label{alg:DAO-2locations}
\begin{algorithmic}[1]
\Require $\APP$ with threshold $\alpha$, and $\ORD$ information; facilities $x$ and $y$ with capacities $\capacity_x$ and $\capacity_y$.
\Ensure Matching $\matching$.
\State Partition the agents in $\Agents_x^1, \Agents_x^2, \Agents_y^1, \Agents_y^2$.
\State $n_x \gets |\Agents_x^1 \cup \Agents_x^2,|$ and $n_y \gets |\Agents_y^1 \cup \Agents_y^2|$.
\State $w \gets \arg\max_{z\in\{x,y\}}\{n_z - c_z\}$ and $z \gets \{x,y\}\setminus\{w\}$.
\State Define agent ordering $\pi_\Agents \gets \Agents_{w}^1, \Agents_{w}^2, \Agents_{z}^1, \Agents_{z}^2$.
\State Define facility ordering $\pi_\Facilities \gets w \succ z$.
\For{each agent $i \in \pi_\Agents$}
   \State $\IndividualMatching_i \gets $ next non-depleted facility in $\pi_\Facilities$.
\EndFor
\State \Return $\matching = (\IndividualMatching_i)_{i \in N}$.
\end{algorithmic} 
\end{algorithm}

\begin{theorem}\label{thm:distances+approvals+ordinal:two-locations:parameterized-upper}
The distortion of {\sc $\alpha$-\AOLAlgTwo} is at most $\max\left\{ 1 + \frac{2}{\alpha}, \frac{3\alpha-1}{\alpha+1} \right\}$.
\end{theorem}

\begin{proof}
Without loss of generality, suppose that $x$ is located at $0$, $y$ is located at $1$, and $w=x$ is the over-demanded facility to which our algorithm gives priority. 
To simplify our discussion in this proof, we will refer to the facilities by their locations, $0$ and $1$. 
Let $\matching = (\IndividualMatching_i)_i$ be the matching computed by the algorithm, and denote by $\OptMatching = (\IndividualOpt_i)_i$ an optimal matching that assigns the agents to the facilities from left to right on the line (such an optimal matching exists due to \cref{lem:line:optimal}). 
Since $c_0+c_1=n$, for any agent $i$ that is not optimally assigned to a facility, i.e., $\IndividualMatching_i \neq \IndividualOpt_i$, there must exist an agent $j$ such that $\IndividualMatching_j = \IndividualOpt_i$ and $\IndividualMatching_i = \IndividualOpt_j$. 
For any such pair of agents $(i,j)$, we will show that
\begin{align*}
    d(i,\IndividualMatching_i) + d(j,\IndividualMatching_j) \leq \max\left\{ 1 + \frac{2}{\alpha}, \frac{3\alpha-1}{\alpha+1} \right\} \cdot \bigg( d(i,\IndividualOpt_i) + d(j,\IndividualOpt_j) \bigg) 
\end{align*}
and the statement will then follow by summing over all agents. 

Without loss of generality, suppose that $\IndividualOpt_i = 0$ and $\IndividualOpt_j = 1$, but $\IndividualMatching_i=1$ and $\IndividualMatching_j=0$. Since $\IndividualOpt_i < \IndividualOpt_j$, by \cref{lem:line:optimal}, it must be the case that $i$'s location on the line is (weakly) to the left of $j$'s location. 
Since agents are first repeatedly assigned to facility $0$ until its capacity is depleted, it must be the case that $j$ was ranked before $i$ in $\pi_\Agents$ (i.e., $j$ is matched to a facility before $i$ is matched to a facility). Since $0$ is over-demanded, its capacity is depleted before the agents that prefer it are depleted, and we thus have that $j \in \Agents_0^1 \cup \Agents_0^2$. Since $i$ is (weakly) to the left of $j$, the same must be true for $i$, that is, $i \in \Agents_0^1 \cup \Agents_0^2$.
Consequently, given that agents in $\Agents_0^1$ have priority over the agents in $\Agents_0^2$, 
there are two cases to consider: 
{\bf (case 1)} $j \in \Agents_0^1$ and $i \in \Agents_0^1 \cup \Agents_0^2$, and
{\bf (case 2)} $j \in \Agents_0^2$ and $i \in \Agents_0^2$.

\medskip
\noindent 
{\bf (case 1): $j \in \Agents_0^1$ and $i \in \Agents_0^1 \cup \Agents_0^2$.}
Since $j \in \Agents_0^1$ and $\IndividualMatching_j = 0$, we have that $A_j = \{\IndividualMatching_j\}$, and thus $d(j,\IndividualOpt_j) > \alpha \cdot d(j,\IndividualMatching_j) \Leftrightarrow d(j,\IndividualMatching_j) < \frac{1}{\alpha} d(j,\IndividualOpt_j)$. 
Using this and the triangle inequality, we have
\begin{align*}
    d(i,\IndividualMatching_i) + d(j,\IndividualMatching_j) 
    &\leq d(i,\IndividualOpt_i) + d(\IndividualMatching_i,\IndividualOpt_i) + d(j,\IndividualMatching_j) \\
    &= d(i,\IndividualOpt_i) + d(\IndividualMatching_j,\IndividualOpt_j) + d(j,\IndividualMatching_j) \\
    &\leq d(i,\IndividualOpt_i) + d(j,\IndividualOpt_j) + 2d(j,\IndividualMatching_j) \\
    &\leq d(i,\IndividualOpt_i) + d(j,\IndividualOpt_j) + \frac{2}{\alpha} d(j,\IndividualOpt_j) \\
    &\leq \left( 1 + \frac{2}{\alpha} \right) \cdot \bigg( d(i,\IndividualOpt_i) + d(j,\IndividualOpt_j)  \bigg).
\end{align*}

\medskip
\noindent 
{\bf (case 2): $j \in \Agents_0^2$ and $i \in \Agents_0^2$.}
Since $A_i = A_j = \{0,1\}$ and both agents prefer $0$ over $1$, there are two possible intervals in which they may be located in: $\left(-\infty,-\frac{1}{\alpha-1}\right]$ or $\left[\frac{1}{\alpha+1}, \frac12 \right]$. We consider the following exhaustive cases.
\begin{itemize}
    \item $i, j \in \left(-\infty,-\frac{1}{\alpha-1}\right]$. Then, we clearly have that 
    \begin{align*}
        d(i,\IndividualMatching_i) + d(j,\IndividualMatching_j) 
        &= d(i,1) + d(j,0) \\
        &= d(i,j) + d(j,1) + d(j,0) 
        = d(i,0) + d(j,1) 
        = d(i,\IndividualOpt_i) + d(j,\IndividualOpt_j). 
    \end{align*}

    \item $i \in \left(-\infty,-\frac{1}{\alpha-1}\right], j \in \left[\frac{1}{\alpha+1}, \frac12 \right]$.
    Then,
    \begin{align*}
        &d(i,\IndividualMatching_i) 
        = d(i,1) = d\left(i, -\frac{1}{\alpha-1}\right) + 1 + \frac{1}{\alpha-1} 
        = d\left(i, -\frac{1}{\alpha-1}\right) + \frac{\alpha}{\alpha-1}; \\
        &d(i,\IndividualOpt_i) 
        = d(i,0) = d\left(i, -\frac{1}{\alpha-1}\right) + \frac{1}{\alpha-1}; \\
        &d(j,\IndividualMatching_j) 
        = d(j,0) \leq \frac12; \\
        &d(j,\IndividualOpt_j) 
        = d(j,1) \geq \frac12.
    \end{align*}
    By appropriately applying the inequality $\frac{a+b}{c+b} \leq \frac{a}{c}$ for any $a \geq c$, with $b = d\left(i, -\frac{1}{\alpha-1}\right)$, we have that
    \begin{align*}
        \frac{d(i,\IndividualMatching_i) + d(j,\IndividualMatching_j)}{d(i,\IndividualOpt_i) + d(j,\IndividualOpt_j)} \leq \frac{d\left(i, -\frac{1}{\alpha-1}\right) + \frac{\alpha}{\alpha-1} + \frac12}{d\left(i, -\frac{1}{\alpha-1}\right) + \frac{1}{\alpha-1} + \frac12} 
        \leq \frac{3\alpha-1}{\alpha+1}. 
    \end{align*}

    \item $i, j \in \left[\frac{1}{\alpha+1}, \frac12 \right]$.
    Then, 
    \begin{align*}
        &d(i,\IndividualMatching_i) 
        = d(i,1) \leq 1- \frac{1}{\alpha+1} = \frac{\alpha}{\alpha+1}; \\
        &d(i,\IndividualOpt_i) 
        = d(i,0) \geq \frac{1}{\alpha+1}; \\
        &d(j,\IndividualMatching_j) 
        = d(j,0) \leq \frac12; \\
        &d(j,\IndividualOpt_j) 
        = d(j,1) \geq \frac12.
    \end{align*}
    Hence, 
     \begin{align*}
        \frac{d(i,\IndividualMatching_i) + d(j,\IndividualMatching_j)}{d(i,\IndividualOpt_i) + d(j,\IndividualOpt_j)} 
        \leq \frac{\frac{\alpha}{\alpha+1} + \frac12}{ \frac{1}{\alpha+1} + \frac12}
        = \frac{3\alpha+1}{\alpha+3} 
        \leq \frac{3\alpha-1}{\alpha+1},
     \end{align*}
\end{itemize}
where the last inequality holds for any $\alpha \geq 1$.
The proof is now complete. 
\end{proof}

By optimizing over $\alpha$, \cref{thm:distances+approvals+ordinal:two-locations:parameterized-upper} yields the following. 

\begin{corollary} \label{cor:distances+approvals+ordinal:two-locations:optimized-upper}
    The distortion of {\sc $(1+\sqrt{2})$-\AOLAlgTwo} is at most $2\sqrt{2}-1 \approx 1.83$.
\end{corollary}

Next, we show the lower bound. We will first prove a parameterized lower bound that will be useful not only for the case of two facilities, but also for more facilities in \cref{sec:distances+approvals+ordinal:lower}. This bound captures algorithms that use large values of $\alpha$ for the $\APP$ information.  

\begin{lemma}\label{lem:distances+approvals+ordinal:lower:large-alpha}
The distortion of any algorithm that uses $\DIST + \APP + \ORD$ information is at least $\frac{3\alpha-1}{\alpha+1}$, even when there $\NumberFacilities=2$ facilities. 
\end{lemma}

\begin{proof}
Consider an instance with two unit-capacity facilities and two agents. Facility $x$ is located at $0$ and facility $y$ is located at $1$. 
The two agents have the same ordinal and approval preferences over the facilities: $x \succ_i y$ and $A_i =\{x,y\}$ for each agent $i \in [2]$. Given this symmetric information, an algorithm may output any of the two possible matchings. So, without loss of generality, suppose that the matching is $\matching = (x,y)$, i.e., agent $1$ is matched to facility $\IndividualMatching_1 = x$, and agent $2$ is matched to facility $\IndividualMatching_2 = y$. 
The agents might be positioned so that agent $1$ is at $1/2$ and agent $2$ is at $-\frac{1}{\alpha-1}$. This metric space is illustrated in \cref{fig:distances+approvals+ordinal:lower:large-alpha}. 

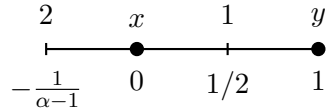
\begin{figure}[h!]
\centering
\begin{tikzpicture}[scale=0.8]
  \draw[thick] (-1.5,0) -- (3,0);

  \draw[thick,fill=black] (0,0) circle (0.1);
  \node[below] at (0,-0.2) {$0$};
  \node[above] at (0,0.2) {$x$};

  \draw[thick,fill=black] (3,0) circle (0.1);
  \node[below] at (3,-0.2) {$1$};
  \node[above] at (3,0.2) {$y$};

  \draw[thick] (1.5,0.1) -- (1.5,-0.1);
  \node[below] at (1.5,-0.2) {$1/2$};
  \node[above] at (1.5,0.25) {$1$};

  \draw[thick] (-1.5,0.1) -- (-1.5,-0.1);
  \node[below] at (-1.5,-0.2) {$-\frac{1}{\alpha-1}$};
  \node[above] at (-1.5,0.25) {$2$};
\end{tikzpicture}
\caption{The metric space used in the proof of \cref{lem:distances+approvals+ordinal:lower:large-alpha}.}
\label{fig:distances+approvals+ordinal:lower:large-alpha}
\end{figure}

Observe that these agent positions are consistent to their ordinal preferences, and also to their approval preferences since $d(1,y) = 1/2 = d(1,x)$ and $d(2,y) = 1+ \frac{1}{\alpha-1} = \frac{\alpha}{\alpha-1} = \alpha \cdot d(2,x)$. The social cost of $\matching$ is 
\begin{align*}
    \SC(\matching)=d(1,x)+d(2,y)=\frac12+\frac{\alpha}{\alpha-1}=\frac{3\alpha-1}{2(\alpha-1)}.
\end{align*}
The optimal matching is $\OptMatching = (y,x)$, i.e., agent $1$ is matched to facility $y$ and agent $2$ is matched to facility $x$, with a social cost of
\begin{align*}
    \SC(\OptMatching) = d(1,y) + d(2,x) = \frac12 + \frac{1}{\alpha-1} = \frac{\alpha+1}{2(\alpha-1)}.
\end{align*}
Therefore, the distortion is at least 
\begin{align*}
    \frac{\SC(\matching)}{\SC(\OptMatching)} = \frac{3\alpha-1}{\alpha+1},
\end{align*}
as desired. 
\end{proof}

Using \cref{lem:distances+approvals+ordinal:lower:large-alpha} for large values of $\alpha$, and another, appropriately tuned instance for small values of $\alpha$, we obtain the following.

\begin{theorem}\label{thm:distances+approvals+ordinal:two-locations:lower}
The distortion of any algorithm that uses $\DIST + \APP + \ORD$ information is at least $2\sqrt{2} - 1$, even when there are $\NumberFacilities=2$ facilities.
\end{theorem}

\begin{proof}
If $\alpha \geq 1+\sqrt{2}$, then, by \cref{lem:distances+approvals+ordinal:lower:large-alpha}, the distortion is at least 
\begin{align*}
    \frac{3\alpha-1}{\alpha+1} \geq \frac{3(1+\sqrt{2})-1}{1+\sqrt{2}+1} = 2\sqrt{2}-1.
\end{align*}
So, it suffices to consider the case $\alpha < 1+\sqrt{2}$. Consider an instance with two unit-capacity facilities $x$ and $y$ located at $0$ and $1$, respectively. The two agents prefer $x$ over $y$, but only approve $x$. In other words, $x \succ_i y$ and $A_i =\{x\}$ for each agent $i \in [2]$. This information is symmetric and there is no way of distinguishing which of the two possible matchings is better. Suppose then without loss of generality that the algorithm outputs the matching $\matching = (x,y)$. In this case, the agents might be positioned so that agent $1$ is at $\frac{1}{1+\alpha}-\varepsilon$, where $\varepsilon$ is an infinitesimal, and agent $2$ is at $0$. This metric space is illustrated in \cref{fig:distances+approvals+ordinal:lower:2locations:small-alpha}.

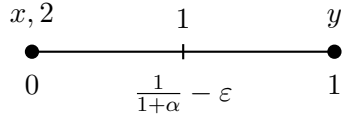
\begin{figure}[t]
\centering
\begin{tikzpicture}[scale=1]
\draw[thick] (0,0) -- (4,0);
\draw[thick,fill=black] (0,0) circle (0.08);
\node[below] at (0,-0.18) {$0$};
\node[above] at (0,0.18) {$x, 2$};
\draw[thick] (2,0.1) -- (2,-0.1);
\node[below] at (2,-0.18) {$\frac{1}{1+\alpha}-\varepsilon$};
\node[above] at (2,0.18) {$1$};
\draw[thick,fill=black] (4,0) circle (0.08);
\node[below] at (4,-0.18) {$1$};
\node[above] at (4,0.18) {$y$};
\end{tikzpicture}
\caption{The metric space used in the proof of \cref{thm:distances+approvals+ordinal:two-locations:lower} when $\alpha < 1+\sqrt{2}$.}
\label{fig:distances+approvals+ordinal:lower:2locations:small-alpha}
\end{figure}

Observe that these agent positions are consistent to the ordinal preferences of the agents, and also to their approval preferences since $d(1,y) = 1 - \frac{1}{1+\alpha}+\varepsilon = \frac{\alpha}{1+\alpha} + \varepsilon > \alpha \cdot \bigg( \frac{1}{1+\alpha}-\varepsilon \bigg) = \alpha \cdot d(1,x)$ and $d(2,y) = 1 > \alpha \cdot 0 = \alpha \cdot d(2,x)$. The social cost of $\matching$ is 
\begin{align*}
    \SC(\matching)=d(1,x)+d(2,y) = \frac{1}{1+\alpha} - \varepsilon +1 = \frac{2+\alpha}{1+\alpha} - \varepsilon.
\end{align*}
The social cost of the optimal matching $\OptMatching = (y,x)$ is
\begin{align*}
    \SC(\OptMatching) = d(1,y) + d(2,x) = \frac{\alpha}{1+\alpha} + \varepsilon.
\end{align*}
Therefore, as $\varepsilon$ tends to $0$, and since $\alpha < 1+\sqrt{2}$, the distortion is 
\begin{align*}
    \frac{2+\alpha}{\alpha} = 1 + \frac{2}{\alpha} > 1+ \frac{2}{1+\sqrt{2}} = 2\sqrt{2}-1.
\end{align*}
In any case, the distortion is at least $2\sqrt{2}-1$ and the proof is complete. 
\end{proof}

Together, \cref{cor:distances+approvals+ordinal:two-locations:optimized-upper} and \cref{thm:distances+approvals+ordinal:two-locations:lower} imply that, for two facilities, the best possible distortion achievable with $\APP+\ORD$ distortion is $2\sqrt{2}-1$, and this is tight even for algorithms with access to $\APP+\ORD+\DIST$ information.

\section{More Facilities with Distance, Approval, and Ordinal Information}
\label{sec:>=3-facilities:distances+approvals+ordinal}
In this section, we consider the case of three or more facilities, and show a tight bound of $2$ on the distortion of algorithms that use $\DIST+\APP+\ORD$ information. To establish this bound, in \cref{sec:distances+approvals+ordinal:three-and-same-top}, we first design a $2$-distortion algorithm for instances with three facilities such that all agents rank the same facility first. We then also provide an algorithm that achieves a distortion of $2$ for instances that are ``demand-monotone'' in the sense that all prefixes of facilities are either over-demanded or under-demanded (\cref{sec:distances+approvals+ordinal:monotone-upper}). Using these two algorithms as subroutines, we are then able to define our main algorithm for handling instances with any number of facilities (\cref{sec:distances+approvals+ordinal:general-upper}). The tight lower bound of $2$ is presented in \cref{sec:distances+approvals+ordinal:lower}. 

\subsection{Upper Bound for Three Facilities and Identical Top-Preference Agents} 
\label{sec:distances+approvals+ordinal:three-and-same-top}
Here we focus on instances with three facilities $\Facilities = \{x, y, z\}$ such that $y$ is located between $x$ and $z$ (either of $x$ or $z$ can be to the left of $y$), $d(x,y)\geq d(y,z)$, and every agent prefers the middle facility, i.e., $\favorite_i = y$ for all $i \in \Agents$. We will design an algorithm that achieves distortion at most $2$ using $\DIST$, $\ORD$, and $\APP$ with threshold $\alpha=3$. 

This algorithm  (\cref{alg:3-locations-top-mid}) first partitions the set of agents into $\Agents_x$ and $\Agents_z$, depending on whether the second most-preferred facility (after $y$, which is their top choice) is $x$ or $z$, respectively. Then, these two sets are further partitioned into three subsets each, depending on the size of the approval set. These partitions are ordered, and the agents are assigned in that order first to $x$, until its capacity is depleted, then to $y$, and then to $z$.

\begin{algorithm}
\caption{{\sc \AOLAlgThreeSF}}
\label{alg:3-locations-top-mid}
\begin{algorithmic}[1]
\Require $\DIST$, $\APP$ with threshold $\alpha=3$, and $\ORD$ information for instances with facilities $\{x, y, z\}$ such that $y$ is located between the other two, $d(x,y)\geq d(y,z)$, and $\favorite_i=y$ for all $i \in \Agents$.
\Ensure Matching $\matching$.
\State Partition $N$ into $\Agents_{x} \gets \{i \in \Agents: x \succ_i z\}$ and $\Agents_{z} \gets \{i \in \Agents: z \succ_i x\}$.
\State Partition $\Agents_{x}$ into three disjoint sets:
\[\Agents_{x}^3 \gets \{i\in \Agents_{x} : |\App_i|=3\}, \ \ \Agents_{x}^2  \gets \{i\in \Agents_{x} : |\App_i|=2\}, \ \ \Agents_{x}^1   \gets \{i\in \Agents_{x} : |\App_i|=1\}.\]
\State Partition $\Agents_{z}$ into three disjoint sets:
\[\Agents_{z}^3 \gets \{i\in \Agents_{z} : |\App_i|=3\}, \ \
\Agents_{z}^2  \gets \{i\in \Agents_{z} : |\App_i|=2\}, \ \
\Agents_{z}^1   \gets \{i\in \Agents_{z} : |\App_i|=1\}.
\]
\State Define agent ordering $\pi_\Agents \gets  \Agents_{x}^3 \succ \Agents_{x}^2 \succ \Agents_{x}^1 \succ \Agents_{z}^3 \succ \Agents_{z}^1 \succ \Agents_{z}^2$.
\State Define item ordering $\pi_\Facilities \gets x\succ y\succ z$.
\For{each agent $i \in \pi_\Agents$}
   \State $\IndividualMatching_i \gets $ next non-depleted facility in $\pi_\Facilities$.
\EndFor
\State \Return $\matching = (\IndividualMatching_i)_{i \in \Agents}$.
\end{algorithmic} 
\end{algorithm}

\begin{theorem}\label{thm:3Fdistortion}
    The distortion of {\sc \AOLAlgThreeSF} is at most $2$ for all instances with three facilities $\{x, y, z\}$ such that $y$ is located between the other two, $d(x,y)\geq d(y,z)$, and $\favorite_i=y$ for all $i \in \Agents$.
\end{theorem}

\begin{proof}
We scale and shift the instance so that $y$ is located at $0$, $z$ is located at $1$, and $x$ is at $-\ratio$, where $\ratio=\frac{d(x,y)}{d(y,z)}\geq 1$ quantifies the extent to which the distance of $x$ from $y$ is greater than the distance of $z$ from $y$. 
Since all agents rank $y$ first, they all lie in the interval $[-\frac{\ratio}{2},\frac{1}{2}]$. Within this interval:
\begin{itemize}
    \item An agent located to the left of $-\frac{\ratio}{4}$ approves $x$;
    \item An agent located in $[-\frac{\ratio}{2},-\frac{1}{2}]$ or in $[\frac{1}{4},\frac{1}{2}]$ approves $z$;
    \item An agent to the left of $\frac{1-\ratio}{2}$ prefers $x$ over $z$.
\end{itemize}
The order of these intervals depends on whether $\ratio\leq 2$ or $\ratio>2$, and we thus consider each of these cases separately. 

\medskip
\noindent 
\textbf{Case $\ratio \leq 2$.}
The order of the borders of the intervals defining the preferences of the agents over the facilities induces the following buckets, where each bucket consists of agents with the same approval information and the same ordinal comparison between $x$ and $z$:
\begin{align*}
    \Agents_x^3 &\subseteq \left[-\frac{\ratio}{2},-\frac{1}{2}\right], \ 
    \Agents_x^2 \subseteq \left[-\frac{1}{2},-\frac{\ratio}{4}\right], \
    \Agents_x^1 \subseteq \left[-\frac{\ratio}{4},\frac{1-\ratio}{2}\right], \ 
    \Agents_z^1 \subseteq \left[\frac{1-\ratio}{2},\frac{1}{4}\right], \
    \Agents_z^2 \subseteq \left[\frac{1}{4},\frac{1}{2}\right].
\end{align*}
Observe that the bucket $\Agents_z^3$ is empty: Agents that approve both $x$ and $z$ must lie in
$[-\frac{\ratio}{2},-\frac{1}{2}]$, while agents that prefer $z$ to $x$ must lie after
$\frac{1-\ratio}{2}$. Since $\ratio\leq 2$, we have $-\frac{1}{2}\leq \frac{1-\ratio}{2}$, so these two
regions do not overlap. Hence, no agent both approves $x$ and $z$ and prefers $z$ to $x$.
The resulting partition of the agents is shown in \cref{fig:3-locations-k-leq-2}.
The following lemma will be useful to bound the total cost of the agents within the same bucket. 

\begin{lemma}\label{lem:distances+approvals+ordinal:three:bucket-swap-bound}
    Let $\subsetOfAgents \subseteq \Agents$ be a set of agents that all lie in an interval $[a,b]$ that contains no facility. 
    Suppose that two matchings $\matching$ and $\OptMatching$ assign the same number of agents of $\subsetOfAgents$ to each facility. If $n_L$ of these agents are assigned to facilities to the left of the interval, and $n_R$ of them are assigned to facilities to the right of the interval, then
    \begin{align*}
        \sum_{i\in B}d(i,\IndividualMatching_i)
        \leq
        \sum_{i\in B}d(i,\IndividualOpt_i)+2(b-a)\min\{n_L,n_R\}.
    \end{align*}
\end{lemma}
\begin{proof}
    Since the two matchings assign the same number of agents of $\subsetOfAgents$ to each facility, their difference can be decomposed into swaps among agents of $\subsetOfAgents$. A swap between two facilities on the same side of the interval does not change the total cost. A swap between one facility to the left of the interval and one facility to the right of the interval can increase the total cost by at most $2(b-a)$, since both agents lie in an interval of length $b-a$. There are at most $\min\{n_L,n_R\}$ such swaps, proving the claim.
\end{proof}

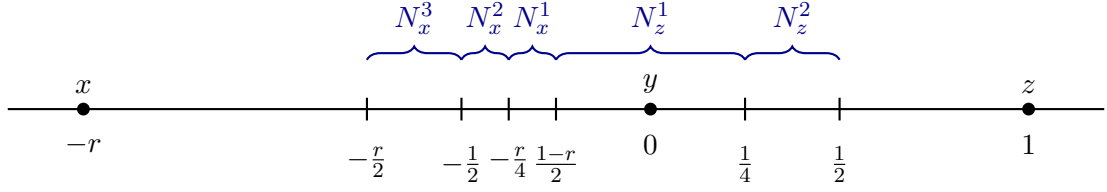
\begin{figure}[t]
    \centering
    \begin{tikzpicture}[scale=1, thick]
        \def\scaleFactor{5.0} 

        \def\xx{-1.5   * \scaleFactor}
        \def\xa{-0.75  * \scaleFactor}
        \def\xb{-0.5   * \scaleFactor}
        \def\xc{-0.375 * \scaleFactor}
        \def\xd{-0.25  * \scaleFactor}
        \def\xy{0.0    * \scaleFactor}
        \def\xe{0.25   * \scaleFactor}
        \def\xf{0.5    * \scaleFactor}
        \def\xz{1.0    * \scaleFactor}

        \draw ({-1.7 * \scaleFactor}, 0) -- ({1.2 * \scaleFactor}, 0);

        \foreach \x/\label in {
            \xa/$-\frac{\ratio}{2}$,
            \xb/$-\frac{1}{2}$,
            \xc/$-\frac{\ratio}{4}$,
            \xd/$\frac{1-\ratio}{2}$,
            \xe/$\frac{1}{4}$,
            \xf/$\frac{1}{2}$
        } {
            \draw (\x, 0.15) -- (\x, -0.15) node[below=0.2cm] {\label};
        }

        \foreach \x/\label in {
          \xx/$-\ratio$,
          \xy/$0$,
          \xz/$1$%
        } {
            \draw (\x, 0) -- (\x, -0) node[below=0.2cm] {\label};
        }

        \draw[fill=black] (\xx,0) circle (2pt) node[above=2pt] {$x$};
        \draw[fill=black] (\xy,0) circle (2pt) node[above=2pt] {$y$};
        \draw[fill=black] (\xz,0) circle (2pt) node[above=2pt] {$z$};

        \draw[decorate, decoration={brace, amplitude=5pt}, DarkBlue, thick]
            (\xa, 0.65) -- (\xb, 0.65)
            node[midway, above=0.2cm] {$\Agents_{x}^{3}$};
        \draw[decorate, decoration={brace, amplitude=5pt}, DarkBlue, thick]
            (\xb, 0.65) -- (\xc, 0.65)
            node[midway, above=0.2cm] {$\Agents_{x}^{2}$};
        \draw[decorate, decoration={brace, amplitude=5pt}, DarkBlue, thick]
            (\xc, 0.65) -- (\xd, 0.65)
            node[midway, above=0.2cm] {$\Agents_{x}^{1}$};
        \draw[decorate, decoration={brace, amplitude=5pt}, DarkBlue, thick]
            (\xd, 0.65) -- (\xe, 0.65)
            node[midway, above=0.2cm] {$\Agents_{z}^{1}$};
        \draw[decorate, decoration={brace, amplitude=5pt}, DarkBlue, thick]
            (\xe, 0.65) -- (\xf, 0.65)
            node[midway, above=0.2cm] {$\Agents_{z}^{2}$};
\end{tikzpicture}
    \caption{Illustration of the possible intervals where agents may lie in when $\ratio \leq 2$ in~\cref{alg:3-locations-top-mid}.}
    \label{fig:3-locations-k-leq-2}
\end{figure}

Let $\OptMatching$ be an optimal matching which, by~\cref{lem:line:optimal}, assigns agents to facilities from left to right. 
The algorithm follows the same left-to-right order of the buckets. Therefore, for each bucket, $\matching$ and $\OptMatching$ assign the same number of agents from the bucket to each facility; the only difference is which agents inside the bucket are assigned to those facilities. 
Hence, the condition of~\cref{lem:distances+approvals+ordinal:three:bucket-swap-bound} holds for each bucket, and we can thus analyze the cost bucket by bucket. Since the buckets form a partition of the agents, and the total cost is additive over agents, it is sufficient to compare the algorithm and the optimal matching inside each bucket separately. In particular, for each bucket $S$ we will show that 
\begin{align*}
    \sum_{i \in S} d(i,\IndividualMatching_i) \leq 2 \sum_{i \in S} d(i,\IndividualOpt_i), 
\end{align*}
and the distortion upper bound of $2$ will then follow by summing over all buckets. 

\begin{itemize}
\item 
\textbf{Agents in \(\Agents_x^3\cup \Agents_x^2\).}
Let \(n_L\) be the number of agents in \(\Agents_x^3\cup \Agents_x^2\) that are matched to \(x\) by \(\bo\), and \(n_R\) the number of agents in \(\Agents_x^3\cup \Agents_x^2\) that are matched to \(y\) or \(z\) by \(\bo\). Since every agent in \(\Agents_x^3\cup \Agents_x^2\) lies in \([-\frac{\ratio}{2},-\frac{\ratio}{4}]\), we have
\begin{align*}
    \sum_{i\in \Agents_x^3\cup \Agents_x^2} d(i,\IndividualOpt_i)
    \geq
    n_L\frac{\ratio}{2}+n_R\frac{\ratio}{4} \geq \frac{\ratio}{2} \min\{n_L,n_R\}.
\end{align*}
The interval containing \(\Agents_x^3\cup \Agents_x^2\) has length \(\frac{\ratio}{4}\), and thus, by \cref{lem:distances+approvals+ordinal:three:bucket-swap-bound},
    \begin{align*}
        \sum_{i\in \Agents_x^3\cup \Agents_x^2} d(i,\IndividualMatching_i)
        \leq \sum_{i\in \Agents_x^3\cup \Agents_x^2} d(i,\IndividualOpt_i) +\frac{\ratio}{2} \min\{n_L,n_R\} 
        \leq 2\sum_{i\in \Agents_x^3\cup \Agents_x^2} d(i,\IndividualOpt_i).
    \end{align*}

\item 
\textbf{Agents in \(\Agents_x^1\).}
Let \(n_L\) be the number of agents in \(\Agents_x^1\) that are matched to \(x\) by \(\bo\), and \(n_R\) the number of agents in \(\Agents_x^1\) that are matched to \(y\) or \(z\) by \(\bo\). Since every agent in \(\Agents_x^1\) lies in \([-\frac{\ratio}{4},\frac{1-\ratio}{2}]\), we have
\begin{align*}
    \sum_{i\in \Agents_x^1} d(i,\IndividualOpt_i)
    \geq  n_L\frac{3\ratio}{4}+n_R\frac{\ratio-1}{2} 
    \geq \frac{5r-2}{2} \min\{n_L,n_R\} 
    \geq \frac{2-r}{2} \min\{n_L,n_R\},
\end{align*}
where the last inequality is true since $r \geq 1$. 
The interval containing \(\Agents_x^1\) has length \(\frac{\ratio}{4}-\frac{\ratio-1}{2} = \frac{2-r}{4}\), and thus, by~\cref{lem:distances+approvals+ordinal:three:bucket-swap-bound},
    \begin{align*}
        \sum_{i\in \Agents_x^1} d(i,\IndividualMatching_i)
        &\leq \sum_{i\in \Agents_x^1} d(i,\IndividualOpt_i)  +\frac{2-r}{2} \min\{n_L,n_R\}
        \leq 2\sum_{i\in \Agents_x^1} d(i,\IndividualOpt_i).
    \end{align*}

\item 
\textbf{Agents in \(\Agents_z^1\).}
Let \(n_L,n_C,n_R\) be the numbers of agents in \(\Agents_z^1\) that are matched by \(\bo\) to \(x,y,z\), respectively. Since every agent in \(\Agents_z^1\) lies in \([\frac{1-\ratio}{2},\frac{1}{4}]\), we have
\begin{align} \label{eq:three:r-leq-two:opt:z-1}
    \sum_{i\in \Agents_z^1} d(i,\IndividualOpt_i)
    \geq
    n_L\frac{\ratio+1}{2}+n_R\frac{3}{4}.
\end{align}
Here we use the same swapping argument as in~\cref{lem:distances+approvals+ordinal:three:bucket-swap-bound}, but we charge the swaps in two parts. We count them once according to whether they change the agent who is matched to $x$, and once according to whether they change the agent who is matched to $z$. There are at most $\min\{n_L,n_C+n_R\}$ swaps of the first type of cost at most $\ratio-1$. Similarly, there are at most $\min\{n_R,n_L+n_C\}$ swaps of the second type of cost at most $\frac{1}{2}$. Note that if a swap exchanges the agents that are matched to $x$ and $z$, then it is counted in both parts, which is necessary because such a swap may increase the cost by both terms. Hence, 
    \begin{align} \label{eq:three:r-leq-two:alg:z-1}
        \sum_{i\in \Agents_z^1} d(i,\IndividualMatching_i)
        &\leq \sum_{i\in \Agents_z^1} d(i,\IndividualOpt_i) +\min\{n_L,n_C+n_R\}(\ratio-1) +\min\{n_R,n_L+n_C\}\frac{1}{2}.
    \end{align}
Since \(\ratio\leq 2\), we have that $\frac{\ratio+1}{2} \geq \ratio -1$, and thus
\begin{align*}
    n_L\frac{\ratio+1}{2}+n_R\frac{3}{4} \geq n_L(\ratio-1)+n_R\frac{1}{2} \geq \min\{n_L,n_C+n_R\}(\ratio-1) +\min\{n_R,n_L+n_C\}\frac{1}{2}.
\end{align*}
By \eqref{eq:three:r-leq-two:opt:z-1} and \eqref{eq:three:r-leq-two:alg:z-1}, we finally obtain
\begin{align*}
    \sum_{i\in \Agents_z^1} d(i,\IndividualMatching_i) \leq 2 \sum_{i\in \Agents_z^1} d(i,\IndividualOpt_i).
\end{align*}

\item 
\textbf{Agents in \(\Agents_z^2\).}
Let \(n_L\) be the number of agents in \(\Agents_z^2\) that are matched to \(x\) or \(y\) by \(\bo\), and \(n_R\) the number of agents in \(\Agents_z^2\) that are matched to \(z\) by \(\bo\). Since every agent in \(\Agents_z^2\) lies in \([\frac{1}{4},\frac{1}{2}]\), we have
\begin{align*}
    \sum_{i\in \Agents_z^2} d(i,\IndividualOpt_i) \geq n_L\frac{1}{4}+n_R\frac{1}{2} \geq \frac12 \min\{n_L,n_R\}.
\end{align*}
The interval containing \(\Agents_z^2\) has length \(\frac{1}{4}\), and thus, by \cref{lem:distances+approvals+ordinal:three:bucket-swap-bound},
    \begin{align*}
        \sum_{i\in \Agents_z^2} d(i,\IndividualMatching_i)
        &\leq  \sum_{i\in \Agents_z^2} d(i,\IndividualOpt_i) +\frac12 \min\{n_L,n_R\}
        \leq 2\sum_{i\in \Agents_z^2} d(i,\IndividualOpt_i).
    \end{align*}
\end{itemize}

\medskip
\noindent 
\textbf{Case $\ratio> 2$.} 
We now have the following buckets:
\begin{align*}
    \Agents_x^3 &\subseteq \left[-\frac{\ratio}{2},\frac{1-\ratio}{2}\right], \
    \Agents_z^3 \subseteq \left[\frac{1-\ratio}{2},-\frac{\ratio}{4}\right], \
    \Agents_z^2 \subseteq \left[-\frac{\ratio}{4},-\frac{1}{2}\right] \cup \left[\frac{1}{4},\frac{1}{2}\right], \
    \Agents_z^1 \subseteq \left[-\frac{1}{2},\frac{1}{4}\right].
\end{align*}
Observe that here the buckets $\Agents_x^1$ and $\Agents_x^2$ are empty:
Since $\ratio>2$, we have that $\frac{1-\ratio}{2}<-\frac{\ratio}{4}$, which implies that any agent who prefers $x$ to $z$ also approves $x$. 
Moreover, such an agent lies to the left of $\frac{1-\ratio}{2}<-\frac{1}{2}$, and therefore approves $z$ as well. 
Hence, every agent who prefers $x$ to $z$ belongs to $\Agents_x^3$. 
The resulting partition is shown in~\cref{fig:3-bucketsNew}. 
Note that the partition of $\Agents_z^2$ into the two sub-buckets $\Agents_z^{2,L} \subseteq \left[-\frac{\ratio}{4},-\frac{1}{2}\right]$ and $\Agents_z^{2,R} \subseteq \left[\frac{1}{4},\frac{1}{2}\right]$ cannot be inferred from the input.

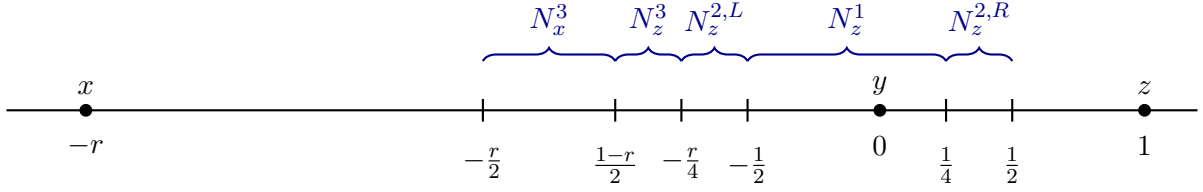
\begin{figure}[t]
    \centering
    \begin{tikzpicture}[scale=1, thick]
        \def\scaleFactor{3.5}

        \def\xx{-3.0   * \scaleFactor}
        \def\xa{-1.5   * \scaleFactor}
        \def\xb{-1.0   * \scaleFactor}
        \def\xc{-0.75  * \scaleFactor}
        \def\xd{-0.5   * \scaleFactor}
        \def\xy{0.0    * \scaleFactor}
        \def\xe{0.25   * \scaleFactor}
        \def\xf{0.5    * \scaleFactor}
        \def\xz{1.0    * \scaleFactor}

        \draw ({-3.3 * \scaleFactor}, 0) -- ({1.2 * \scaleFactor}, 0);

        \foreach \x/\label in {
            \xa/$-\frac{\ratio}{2}$,
            \xb/$\frac{1-\ratio}{2}$,
            \xc/$-\frac{\ratio}{4}$,
            \xd/$-\frac{1}{2}$,
            \xe/$\frac{1}{4}$,
            \xf/$\frac{1}{2}$
        } {
            \draw (\x, 0.15) -- (\x, -0.15) node[below=0.2cm] {\label};
        }

        \foreach \x/\label in {
          \xx/$-\ratio$,
          \xy/$0$,
          \xz/$1$%
        } {
            \draw (\x, 0) -- (\x, -0) node[below=0.2cm] {\label};
        }

        \draw[fill=black] (\xx,0) circle (2pt) node[above=2pt] {$x$};
        \draw[fill=black] (\xy,0) circle (2pt) node[above=2pt] {$y$};
        \draw[fill=black] (\xz,0) circle (2pt) node[above=2pt] {$z$};

        \draw[decorate, decoration={brace, amplitude=5pt}, DarkBlue, thick]
            (\xa, 0.65) -- (\xb, 0.65)
            node[midway, above=0.2cm] {$N^3_{x}$};
        \draw[decorate, decoration={brace, amplitude=5pt}, DarkBlue, thick]
            (\xb, 0.65) -- (\xc, 0.65)
            node[midway, above=0.2cm] {$N^3_{z}$};
        \draw[decorate, decoration={brace, amplitude=5pt}, DarkBlue, thick]
            (\xc, 0.65) -- (\xd, 0.65)
            node[midway, above=0.2cm] {$N^{2,L}_{z}$};
        \draw[decorate, decoration={brace, amplitude=5pt}, DarkBlue, thick]
            (\xd, 0.65) -- (\xe, 0.65)
            node[midway, above=0.2cm] {$N^1_{z}$};
        \draw[decorate, decoration={brace, amplitude=5pt}, DarkBlue, thick]
            (\xe, 0.65) -- (\xf, 0.65)
            node[midway, above=0.2cm] {$N^{2,R}_{z}$};

\end{tikzpicture}
    \caption{Illustration of the possible intervals where agents may lie in when \(\ratio>2\). 
    For clarity, the disconnected pieces of \(N^2_{z}\) are shown as \(N^{2,L}_{z}\) and \(N^{2,R}_{z}\).}
\label{fig:3-bucketsNew}
\end{figure}

Let $\OptMatching$ be an optimal matching assigns the agents from left to right (\cref{lem:line:optimal}). 
In this case, the algorithm does not necessarily follow the left-to-right order of the buckets. 
It first assigns $\Agents_x^3$ and $\Agents_z^3$, which is consistent with their true order on the line, but the remaining buckets $\Agents_z^2$ and $\Agents_z^1$ cannot be treated by the same ordering argument. 
Thus, \cref{lem:distances+approvals+ordinal:three:bucket-swap-bound} does not apply directly for them. 
We therefore need to do the analysis a bit differently. 
We will first bound the cost for the agents in $\Agents_x^3\cup\Agents_z^3$ as in the previous case, and then give a separate bound for the agents in $\Agents_z^2\cup\Agents_z^1$.

\begin{itemize}
\item 
\textbf{Agents in $\Agents_x^3\cup\Agents_z^3$.}
Observe that these agents are contained in the interval $\left[-\frac{\ratio}{2},-\frac{\ratio}{4}\right]$, which is exactly the same
interval that contains the agents of $\Agents_x^3\cup\Agents_x^2$ in the previous case where $\ratio\leq 2$. Since the argument for that case did not dependent on the assumption $\ratio\leq 2$, it can be used verbatim here a well to show that 
\begin{align*}
    \sum_{i\in \Agents_x^3\cup\Agents_z^3} d(i,\IndividualMatching_i) \leq 2\sum_{i\in \Agents_x^3\cup\Agents_z^3} d(i,\IndividualOpt_i).
\end{align*}

\item 
\textbf{Agents in $\Agents_z^2\cup\Agents_z^1$.}
As discussed above, before considering the agents in $\Agents_z^2\cup\Agents_z^1$, both $\matching$ and $\OptMatching$ have assigned the same number of agents to each of $x,y,z$. Hence, the same must be true for $\Agents_z^2\cup\Agents_z^1$.

We construct a directed multigraph $G$ whose vertices are the facilities $x,y,z$. For every agent $i\in \Agents_z^2\cup\Agents_z^1$ with $\IndividualMatching_i\neq\IndividualOpt_i$, we add an edge from $\IndividualMatching_i$ to $\IndividualOpt_i$. Since $\matching$ and $\OptMatching$ assign the same number of agents of $\Agents_z^2\cup\Agents_z^1$ to each facility, every vertex has the same in-degree and out-degree. Hence, $G$ can be decomposed into directed cycles. Since there are only three facilities, each cycle has length $2$ or $3$. We bound the cost on each cycle separately.

\medskip
\noindent
\emph{2-cycles.}
Consider a $2$-cycle formed by agents $i$ and $j$, so that
$\IndividualMatching_i=\IndividualOpt_j$ and $\IndividualMatching_j=\IndividualOpt_i$. First suppose that the two facilities in the cycle are $x$ and $z$. Without loss of generality, let $\IndividualMatching_i=x$ and $\IndividualMatching_j=z$. Since every agent in $\Agents_z^2\cup\Agents_z^1$ lies in $\left[-\frac{\ratio}{4},\frac{1}{2}\right]$, we have
\begin{align*}
    d(i,\IndividualOpt_i)+d(j,\IndividualOpt_j)
    &\geq \frac{3\ratio}{4}+\frac{1}{2} = \frac{3\ratio+2}{4},
\end{align*}
and
\begin{align*}
    d(i,\IndividualMatching_i)+d(j,\IndividualMatching_j) 
    &\leq \left(\ratio+\frac{1}{2}\right)+\left(1+\frac{\ratio}{4}\right) 
    = \frac{5\ratio+6}{4} . 
\end{align*}
Therefore, since $\ratio > 2$, 
\begin{align*}
    \frac{d(i,\IndividualMatching_i)+d(j,\IndividualMatching_j)}{d(i,\IndividualOpt_i)+d(j,\IndividualOpt_j)}
    &\leq \frac{5\ratio+6}{3\ratio+2}
    <2,
\end{align*}
It remains to consider the case where one of the two agents is assigned to $y$. 
Suppose that $\IndividualMatching_i=y$ and $\IndividualMatching_j=z$. 
Consider the two-facility sub-instance consisting only of agents $i$ and $j$, and one unit of capacity for each of $y$ and $z$. 
Since both agents prefer $y$ over $z$, $y$ is the over-demanded facility out of the two. In addition, \cref{alg:3-locations-top-mid} gives priority to the agents of $N_z^1$ over the agents of $N_z^2$. Consequently, for the two-facility sub-instance, \cref{alg:3-locations-top-mid} behaves exactly the same as \cref{alg:DAO-2locations}, and thus leads to the same matching for these two agents. Therefore, by \cref{thm:distances+approvals+ordinal:two-locations:parameterized-upper} with $\alpha=3$, we have that
\begin{align*}
    d(i,\IndividualMatching_i)+d(j,\IndividualMatching_j)
    \leq 2\bigl(d(i,\IndividualOpt_i)+d(j,\IndividualOpt_j)\bigr).
\end{align*}
The case $\IndividualMatching_j=x$ is similar. In the corresponding two-facility sub-instance for $x$ and $y$, 
since these agents approve only $y$ (observe that $i,j \in \Agents_z^2\cup\Agents_z^1$ implies that the agents can approve only $y$ and $z$), 
\cref{alg:DAO-2locations} treats the two agents equally. Therefore, any matching chosen by \cref{alg:3-locations-top-mid} satisfies the same bound that \cref{alg:DAO-2locations} guarantees.

\medskip
\noindent
\emph{3-cycles.}
Now consider a $3$-cycle, and let $i,j,k$ be the involved three agents in their left-to-right order. Since $\OptMatching$ is left-to-right and the cycle uses all three facilities, the optimal matching assigns these agents to $x,y,z$ in this order. There are two possible orientations:
\begin{align*}
    (\IndividualMatching_i=\IndividualOpt_j, \ \IndividualMatching_j=\IndividualOpt_k, \ \IndividualMatching_k=\IndividualOpt_i) \text{ \ and \ }
    (\IndividualMatching_i=\IndividualOpt_k, \ \IndividualMatching_k=\IndividualOpt_j, \ \IndividualMatching_j=\IndividualOpt_i).
\end{align*}
For the first orientation, by the triangle inequality,
\begin{align*}
    d(i,\IndividualMatching_i)+d(j,\IndividualMatching_j)+d(k,\IndividualMatching_k)
    &=
    d(i,\IndividualOpt_j)+d(j,\IndividualOpt_k)+d(k,\IndividualOpt_i)\\
    &\leq
    d(i,\IndividualOpt_j)+d(j,\IndividualOpt_i)+d(k,\IndividualOpt_k)+2d(j,k)\\
    &\leq
    d(i,\IndividualOpt_i)+d(j,\IndividualOpt_j)+d(k,\IndividualOpt_k)
    +2d(i,j)+2d(j,k).
\end{align*}
The other orientation is symmetric and gives the same bound. Since $j$ lies between $i$ and $k$,
\begin{align*}
    d(i,j)+d(j,k)=d(i,k).
\end{align*}
Therefore, in either orientation,
\begin{align*}
    d(i,\IndividualMatching_i)+d(j,\IndividualMatching_j)+d(k,\IndividualMatching_k)
    &\leq
    d(i,\IndividualOpt_i)+d(j,\IndividualOpt_j)+d(k,\IndividualOpt_k)+2d(i,k).
\end{align*}
Since all agents in $\Agents_z^2\cup\Agents_z^1$ lie in $\left[-\frac{\ratio}{4},\frac{1}{2}\right]$,
\begin{align*}
    d(i,k)\leq \frac{\ratio+2}{4}.
\end{align*}
Moreover, since any of these agents is at distance at least $\frac{3\ratio}{4}$ from $x$ and at distance at least $\frac{1}{2}$ from $z$, 
any matching of them to $x,y,z$ has cost at least
\begin{align*}
    \frac{3\ratio}{4}+\frac{1}{2} = \frac{3\ratio+2}{4}.
\end{align*} 
Hence, since $\ratio > 2$, 
\begin{align*}
    \frac{
    d(i,\IndividualMatching_i)+d(j,\IndividualMatching_j)+d(k,\IndividualMatching_k)}{d(i,\IndividualOpt_i)+d(j,\IndividualOpt_j)+d(k,\IndividualOpt_k)}
    &\leq 1+\frac{2d(i,k)}{d(i,\IndividualOpt_i)+d(j,\IndividualOpt_j)+d(k,\IndividualOpt_k)}\\
    &\leq 1+\frac{2\ratio+4}{3\ratio+2}
    <2,
\end{align*}
Finally, summing over all cycles, we obtain
\begin{align*}
    \sum_{i\in \Agents_z^2\cup\Agents_z^1} d(i,\IndividualMatching_i) \leq 2\sum_{i\in \Agents_z^2\cup\Agents_z^1} d(i,\IndividualOpt_i).
\end{align*}
\end{itemize}
The proof is now complete. 
\end{proof}

\subsection{Upper Bound for Demand-Monotone Instances} 
\label{sec:distances+approvals+ordinal:monotone-upper}
We now introduce a subroutine for instances with an arbitrary number of facilities such that the preferences of the agents are {\em demand-monotone}: either every prefix of the facilities is over-demanded, or every prefix is under-demanded. Note that both this algorithm and the algorithm for general instances (presented in the next subsection), use the $\DIST$ information only to extract the true ordering of the facilities on the line. Before presenting the algorithm, we provide some useful exact definitions about the instances we consider here.

\newcommand{\GG}{G}
\begin{definition}
    Let $F=\{x_1, \dots, x_\NumberFacilities\}$ be a set of facilities such that $\loc_{x_1} < \loc_{x_2} < \cdots < \loc_{x_\NumberFacilities}$. For any subset of facilities $\GG \subseteq \Facilities$, let $\capOf{\GG} = \sum_{x \in \GG} \capacity_x$ be the {\bf supply} of $\GG$ (the total capacities of these facilities). Given a set of agents $\Agents$, let $\agentsOf{\GG} = \{i \in \Agents : \favorite_i \in \GG\}$ denote the agents whose top-ranked facility in $\Facilities$ is also in $\GG$; we refer to this set as the {\bf demand} for $\GG$. Also, let $\Facilities_t = \{x_1, \dots, x_t\}$ denote the first (according to their locations) $t$ facilities, with $\Facilities_0=\varnothing$.
\end{definition}

\begin{definition}\label{def:polarized}
Let $F=\{x_1, \dots, x_\NumberFacilities\}$ be a set of facilities such that $\loc_{x_1} < \loc_{x_2} < \cdots < \loc_{x_\NumberFacilities}$, and $\Agents$ be a set of agents. We call the set of facilities $\Facilities$ \textbf{over-demanded} (by $\Agents$) if for every $t\in \{1, 2, \dots, \NumberFacilities\}$, $|\agentsOf{\Facilities_t}| \geq \capOf{\Facilities_t}$, that is, for every prefix of these facilities, the demand weakly exceeds the supply. 
Similarly, we call $F$ \textbf{under-demanded} (by $\Agents$) if $|\agentsOf{\Facilities_t}| \leq \capOf{\Facilities_t}$ for every $t\in \{1, 2, \dots, \NumberFacilities\}$. 
Finally, if $F$ is either over-demanded or under-demanded, we refer to it (and the corresponding instance) as \textbf{demand-monotone}.
\end{definition}

Our $\alpha$-\AOLAlgPolarized\  algorithm (\cref{alg:polarized-match}) takes as input a demand-monotone instance. If it is an under-demanded instance, the algorithm just reverses the order of the facilities, transforming it into an over-demanded one. 
The algorithm identifies the leftmost facility $x_t\in F$ with available capacity and, if it is the only such facility, it assigns all agents to it and returns that matching. Otherwise, it determines which agents to assign to $x_t$ in order to deplete $x_t$'s capacity, and then recursively calls itself to deal with the remaining over-demanded instance.

To determine how to deplete the capacity of $x_t$, the algorithm starts by considering $\agentsOf{F_{t-1}}$, the set of agents whose favorite facility is to the left of $x_t$. If any such agents exist, it assigns a subset of them to $x_t$ (prioritizing agents whose top facility is further to the left) until either the capacity of $x_t$ is depleted or all of the agents in $\agentsOf{F_{t-1}}$ are assigned to it. If the capacity of $x_t$ remains positive, the algorithm then decides which subset of the agents in $\agentsOf{\{x_t\}}$ to assign to $x_t$. Note that, since the instance is over-demanded, $|\agentsOf{F_t}|\geq c_{x_t}$, so we know that there will be enough agents in $\agentsOf{\{x_t\}}$ to deplete the capacity of $x_t$. To determine which subset of these agents to assign to $x_t$, we simulate the $\alpha$-\AOLAlgTwo\  algorithm with the same $\alpha$ as in the input of $\alpha$-\AOLAlgPolarized. Specifically, $\alpha$-\AOLAlgTwo \  is given as input facilities $x_t$ and $x_{t+1}$ with the capacity of $x_t$ being its remaining capacity and the capacity of $x_{t+1}$ being $|\agentsOf{F_t}|- c_{x_t}$, i.e., the number of agents from $\agentsOf{\{x_t\}}$ that will need to be matched to a facility to the right of $x_t$. Note that the true capacity of $x_{t+1}$ may be greater or smaller than that, but the purpose of simulating $\alpha$-\AOLAlgTwo\  is just to determine which set of agents from $\agentsOf{\{x_t\}}$ should be assigned to $x_t$. The algorithm assigns to $x_t$ the same agents as $\alpha$-\AOLAlgTwo, but it does not yet assign any agents to $x_{t+1}$. 

Once the capacity of $x_t$ is depleted, $\alpha$-\AOLAlgPolarized\ is called recursively with updated capacities $c''$ (where $x_t$ has capacity $0$ while all other facilities have the same capacity as before) and the remaining set of agents $N''$ (that ones not assigned to $x_t$) as input. The assignment $\matching''$ returned by this recursive call is combined with the assignment $\matching'$ for $x_t$ described above, and their combination, $\matching$, is the final output of the algorithm.

\begin{algorithm}
\caption{$\alpha$-\AOLAlgPolarized}
\label{alg:polarized-match}
\begin{algorithmic}[1]
\Require A demand-monotone instance with facilities $F$ ordered so that they satisfy $\ell_{x_1}<\dots<\ell_{x_f}$, capacities $(c_x)_{x\in F}$, and agents $\Agents$ that provide $\ORD$ and $\APP$ information for some threshold $\alpha$.
\Ensure Matching $\matching$.
\If{the instance is under-demanded}
    \State Reverse the order of the facilities so that it becomes over-demanded.
\EndIf
\State $x_t \gets$ leftmost facility in $\Facilities$ such that $c_{x_t}>0$.
\If{$x_t$ is the only facility in $\Facilities$ with positive capacity}
    \State \Return $\matching$ with $\IndividualMatching_i \gets x_t$ for all $i\in N$.
\EndIf
\State $\pi_\Agents \gets$ any ordering of $\Agents$ such that $i$ precedes $j$ when $\ell_{\favorite_i} < \ell_{\favorite_j}$.\label{line:ordering}
\If{$|\agentsOf{F_{t-1}}|\geq c_{x_t}$ \label{line:depleted}}
    \State $B \gets $ the first $c_{x_t}$ agents in $\agentsOf{F_{t-1}}$ according to $\pi_\Agents$.
    \State $\IndividualMatching'_i \gets x_t$ for all $i\in B$, and $\Agents'' \gets \Agents \setminus B$.
\Else
    \State $c'_{x_t} \gets c_{x_t}-|\agentsOf{F_{t-1}}|$ and $c'_{x_{t+1}}\gets |\agentsOf{F_{t}}|-c_{x_t}$. 
    \State Simulate $\alpha$-\AOLAlgTwo$((x_t, x_{t+1}), (c'_{x_t}, c'_{x_{t+1}}), \agentsOf{\{x_t\}})$. \label{line:two-agent-call}
    \State $B_t \gets$ agents in $\agentsOf{\{x_t\}}$ that $\alpha$-\AOLAlgTwo\  would assign to $x_t$.
    \State $\IndividualMatching'_i \gets x_t$ for all $i\in B_t\cup  \agentsOf{F_{t-1}}$, and $\Agents'' \gets \Agents \setminus (B_t \cup \agentsOf{F_{t-1}})$.
\EndIf
\State $c''_{x_t} \gets 0$ and $c''_x \gets c_x$ for all $x \in \Facilities\setminus \{x_t\}$.
\State $\matching''\gets \alpha$-\AOLAlgPolarized$(F, c'', N'')$.
\State Combine $\matching'$ with $\matching''$ into $\matching$.
\State \Return $\matching$
\end{algorithmic}
\end{algorithm}

\begin{theorem} \label{thm:polarized-match}
    The distortion of {\normalfont $\alpha$-\AOLAlgPolarized} is at most $\max\left\{ 1 + \frac{2}{\alpha}, \frac{3\alpha-1}{\alpha+1} \right\}$ for any demand-monotone instance. Specifically, for $\alpha=3$, this yields distortion $2$ for any demand-monotone instance.
\end{theorem}
\begin{proof}
    Since the algorithm transforms under-demanded instances into over-demanded ones, it suffices to prove the distortion upper bound for over-demanded instances. The $\alpha$-\AOLAlgPolarized\  algorithm iterates over the facilities from left to right, and assigns agents to them using the ordering $\pi_\Agents$, which also aims to prioritize agents that are further to the left. Specifically, for all $t\in \{1, 2, \dots, \NumberFacilities\}$, the algorithm does not assign any agent from $\agentsOf{\{x_t\}}$ before all agents in $\agentsOf{\Facilities_{t-1}}$ (if any) have been assigned. 

    Let $\OptMatching$ be an optimal matching that greedily matches the leftmost agent to the leftmost facility (its existence is guaranteed due to \cref{lem:line:optimal}). Since every agent in $\agentsOf{\{x_{t-1}\}}$ is to the left of any agent in $\agentsOf{\{x_t\}}$, like $\alpha$-\AOLAlgPolarized, $\OptMatching$ assigns agents in the order $\agentsOf{\{x_1\}}, \agentsOf{\{x_2\}}, \dots, \agentsOf{\{x_
    \NumberFacilities\}}$, but it also ensures that the agents \emph{within} each set $\agentsOf{\{x_t\}}$ are assigned from left to right. The fact that we cannot fully recover the order within these sets using $\DIST+\APP+\ORD$ is precisely the source of the distortion.

    A direct implication of the above is that both $\OptMatching$ and $\matching$ distribute each group $\agentsOf{\{x_t\}}$ over the same set of facilities. Specifically, for each group $\agentsOf{\{x_t\}}$, both $\OptMatching$ and $\matching$ assign exactly the same number of agents from that group to each facility in $\Facilities$, i.e.,
    \begin{equation}\label{eq:same_distribution}
        \forall t\in \{1, 2, \dots, \NumberFacilities\} ~~\text{and}~~ x\in F,  \qquad \left|\{i\in \agentsOf{\{x_t\}}:~\IndividualMatching_i = x\}\right| = \left|\{i\in \agentsOf{\{x_t\}}:~\IndividualOpt_i = x\}\right|.
    \end{equation}
    Therefore, in order to show that the total cost of a set $\agentsOf{\{x_t\}}$ in $\matching$ is not much higher than the total cost of $\agentsOf{\{x_t\}}$ in the optimal matching $\OptMatching$, it suffices to show that there is no alternative assignment \emph{using the exact same capacities as $\matching$} that performs much better. 
    In the rest of the proof, we will show that for each $t\in \{1, 2, \dots, \NumberFacilities\}$, we have
    \begin{equation}\label{eq:local_dist}
        \sum_{i\in \agentsOf{\{x_t\}}} d(i, \IndividualMatching_i) \leq \max\left\{ 1 + \frac{2}{\alpha}, \frac{3\alpha-1}{\alpha+1} \right\} \sum_{i\in \agentsOf{\{x_t\}}} d(i, \IndividualOpt_i).
    \end{equation}
    The statement will then follow by summing over all $t$. 

    Each call to \AOLAlgPolarized\  depletes the capacity of the leftmost non-depleted facility, $x_t$, and it is then called recursively in order to deplete the capacities of subsequent facilities. For any such $t\in \{1, 2, \dots, \NumberFacilities\}$ we consider two cases, depending on the relative size of $|\agentsOf{\Facilities_{t-1}}|$ and $c_{x_t}$:

    \medskip
    \noindent 
    {\bf (case 1) $|\agentsOf{\Facilities_{t-1}}|\geq c_{x_t}$.}
     Then, $\matching$ assigns all of the agents in $\agentsOf{\{x_t\}}$ to facilities that are weakly to the right of $x_{t+1}$. To verify this, consider the case where $c_{x_t}>0$, i.e., facility $x_t$ initially has positive capacity. Then, $x_t$ will be considered by the \AOLAlgPolarized\  algorithm in one of its calls. During that call, since $|\agentsOf{\Facilities_{t-1}}|\geq c_{x_t}$, the condition of~\cref{line:depleted} will be satisfied, so the capacity of $x_t$ will be fully depleted by agents from $\agentsOf{\Facilities_{t-1}}$. Therefore, no agents in $\agentsOf{\{x_t\}}$ will be assigned to $x_t$, and they will instead be assigned to some facility weakly to the right of $x_{t+1}$. Since these agents prefer $x_t$ over $x_{t+1}$, they are all strictly to the left of $x_{t+1}$, so it is easy to see that their total cost is not affected by who is matched to each facility. In fact, their total cost is the same as their optimal total cost in $\OptMatching$: 
    \begin{align*}
    \sum_{i\in \agentsOf{\{x_t\}}} d(i,\IndividualMatching_i) &= \sum_{i\in \agentsOf{\{x_t\}}} d(i, x_{t+1}) + \sum_{i\in \agentsOf{\{x_t\}}} d(x_{t+1},\IndividualMatching_i) \\
    &= \sum_{i\in \agentsOf{\{x_t\}}} d(i, x_{t+1}) + \sum_{i\in \agentsOf{\{x_t\}}} d(x_{t+1},\IndividualOpt_i) \\
    &= \sum_{i\in \agentsOf{\{x_t\}}} d(i,\IndividualOpt_i),
    \end{align*}
    where the first equality uses the fact that each $i\in \agentsOf{\{x_t\}}$ is to the left of $x_{t+1}$, while the facility $\IndividualMatching_i$ for each agent $i$ is weakly to the right of $x_{t+1}$, so $d(i,\IndividualMatching_i)=d(i,x_{t+1})+d(x_{t+1},\IndividualMatching_i)$. The second equation uses the fact that $\matching$ and $\OptMatching$ assign the agents to the same facilities (\cref{eq:same_distribution}), so the total distance of these assignments from $x_{t+1}$ is the same. Note that we have so far assumed that $c_{x_t}>0$, but the same argument applies for the case $c_{x_t}=0$, since the algorithm will not even consider $x_t$ and the agents in $\agentsOf{\{x_t\}}$ will, again, all be matched to facilities weakly to the right of $x_{t+1}$.

    \medskip
    \noindent 
    {\bf (case 2) $|\agentsOf{\Facilities_{t-1}}|< c_{x_t}$.}
    Then, $\matching$ assigns some of the agents in $\agentsOf{\{x_t\}}$ to $x_t$ and the rest of them (if any) to facilities weakly to the right of $x_{t+1}$. To verify this, note that when \AOLAlgPolarized\  considers $x_t$, since $|\agentsOf{\Facilities_{t-1}}|< c_{x_t}$, the condition of~\cref{line:depleted} will be violated, so $c_{x_t}-|\agentsOf{\Facilities_{t-1}}|$ agents from $\agentsOf{\{x_t\}}$ will be assigned to $x_t$ after simulating $\alpha$-\AOLAlgTwo\  in \cref{line:two-agent-call}. Any agents from $\agentsOf{\{x_t\}}$ that are not assigned to $x_t$ are assigned to facilities weakly to the right of $x_{t+1}$. However, the input provided to $\alpha$-\AOLAlgTwo\  assumes that any agents from $\agentsOf{\{x_t\}}$ that are not assigned to $x_t$ can all be assigned to $x_{t+1}$. Specifically, the capacity of $x_{t+1}$ is set to $c'_{x_{t+1}}\gets |\agentsOf{F_{t}}|-c_{x_t}$, which may be more than its true capacity. Nevertheless, the rest of the proof shows that the guidance provided by $\alpha$-\AOLAlgTwo\  allows us to achieve the desired distortion.
        
    Let $B_t$ and $B_{t+1}$ be the sets of agents from $\agentsOf{\{x_t\}}$ that the simulation from \cref{line:two-agent-call} would assign to facilities $x_t$ and $x_{t+1}$, respectively. Note that $B_t$ and $B_{t+1}$ form a partition of $\agentsOf{\{x_t\}}$. Also, let $O_t:=\{i\in \agentsOf{\{x_t\}}~|~ o_i = x_t\}$ be the set of agents that $\OptMatching$ assigns to $x_t$, and $O_{t+1}:= \agentsOf{\{x_t\}}\setminus O_t$ be all the other agents in $\agentsOf{\{x_t\}}$ that $\OptMatching$ assigns to facilities weakly to the right of $x_{t+1}$. Note that since $\matching$ and $\OptMatching$ assign the same number of agents from $\agentsOf{\{x_t\}}$ to each facility, then $|O_t|=|B_t|$ (and thus also $|O_{t+1}|=|B_{t+1}|$). If we let $\rho(\alpha):= \max\left\{ 1 + \frac{2}{\alpha}, \frac{3\alpha-1}{\alpha+1} \right\}$, then by \cref{thm:distances+approvals+ordinal:two-locations:parameterized-upper}, we get
    \begin{equation}\label{eq:approx_guarantee}
        \sum_{i\in B_t} d(i, x_t) + \sum_{i\in B_{t+1}} d(i, x_{t+1}) \leq \rho(\alpha)\cdot \left( \sum_{i\in O_t} d(i, x_t) + \sum_{i\in O_{t+1}} d(i, x_{t+1})\right).
    \end{equation}
    Since $\matching$ does assign $B_t$ to $x_t$, but assigns $B_{t+1}$ to facilities weakly (and maybe strictly) to the right of $x_{t+1}$ we can conclude that the total cost of $\matching$ is
    \begin{align*}
    \sum_{i\in \agentsOf{\{x_t\}}} d(i,\IndividualMatching_i) &= \sum_{i\in B_t} d(i, x_t) +
    \sum_{i\in B_{t+1}}d(i,\IndividualMatching_i)\\
    &=\sum_{i\in B_t} d(i, x_t) + \sum_{i\in B_{t+1}} d(i, x_{t+1}) + \sum_{i\in B_{t+1}} d(x_{t+1},\IndividualMatching_i) \\
    &\leq \rho(\alpha)\cdot \left( \sum_{i\in O_t} d(i, x_t) + \sum_{i\in O_{t+1}} d(i, x_{t+1})\right) + \sum_{i\in B_{t+1}} d(x_{t+1},\IndividualMatching_i) \\
     &=\rho(\alpha)\cdot \left( \sum_{i\in O_t} d(i, x_t) + \sum_{i\in O_{t+1}} d(i, x_{t+1})\right) + \sum_{i\in O_{t+1}} d(x_{t+1},\IndividualOpt_i) \\
    &\leq  \rho(\alpha)\cdot \sum_{i\in \agentsOf{\{x_t\}}} d(i,\IndividualOpt_i),
    \end{align*}
    where the second equation uses the fact that the agents in $B_{t+1}$ are to the left of $x_{t+1}$ but are matched to facilities $\IndividualMatching_i$ to the right of $x_{t+1}$, the inequality is due to \cref{eq:approx_guarantee}, and the subsequent equation uses the fact that, from \cref{eq:same_distribution}, both $\matching$ and $\OptMatching$ assign the agents to the same facilities weakly to the right of $x_{t+1}$, so $\sum_{i\in B_{t+1}} d(x_{t+1},\IndividualMatching_i)= \sum_{i\in O_{t+1}} d(x_{t+1},\IndividualOpt_i)$.
\end{proof}

\subsection{Upper Bound for General Instances}
\label{sec:distances+approvals+ordinal:general-upper}
We are now ready to present our main algorithm which works for arbitrary instances using $\DIST+\ORD+\APP$ information. This algorithm, \GenAlgorithm\ (\cref{alg:general-locations}), starts by identifying the longest prefix $F_t$ of $F$ that is demand-monotone with respect to the preferences of the agents in $N$. If $F_t$ is such that $|\agentsOf{\Facilities_t}| = \capOf{\Facilities_t}$, i.e., the total demand and total supply happen to be equal ($F_t$ could be either over-demanded or under-demanded), then the $\alpha$-\AOLAlgPolarized\  subroutine is called to determine an assignment $\matching'$ of the agents in $\agentsOf{F_t}$ to the facilities in $F_t$. If, on the other hand, this is not the case, then either $F_t$ is over-demanded and $|\agentsOf{\Facilities_t}| > \capOf{\Facilities_t}$, or it is under-demanded and $|\agentsOf{\Facilities_t}| < \capOf{\Facilities_t}$.

For the over-demanded case with $|\agentsOf{\Facilities_t}| > \capOf{\Facilities_t}$, \GenAlgorithm\ identifies the amount of excess demand $e$ in $F_t$. Note that, while $\Facilities_t$ is over-demanded, $\Facilities_{t+1}$ is not over-demanded (otherwise \GenAlgorithm\ would have chosen $\Facilities_{t+1}$ instead of $\Facilities_t$ in~\cref{step:maximality}). This implies that $|\agentsOf{\Facilities_{t+1}}| < \capOf{\Facilities_{t+1}}$, so the capacity of $x_{t+1}$ is more than $e$. Using the agents in $\agentsOf{\Facilities_t}$ and the full capacities of the facilities in $F_t$ along with just $e$ units of capacity from $x_{t+1}$, the algorithm creates a demand-monotone instance and calls $\alpha$-\AOLAlgPolarized\  to compute a matching $\matching'$ for it.

For the under-demanded case with $|\agentsOf{\Facilities_t}| < \capOf{\Facilities_t}$, \GenAlgorithm\ identifies the amount of excess supply $e$ in $F_t$. It first creates a demand-monotone instance by reducing the capacity of $x_t$ by $e$ and calls $\alpha$-\AOLAlgPolarized\  to determine a matching for this instance. Then, the algorithm needs to deal with the main complication that there remains unassigned capacity of $e$ at $x_t$ with no remaining demand in $\agentsOf{F_t}$. On the other hand, note that the demand in $\agentsOf{\{x_{t+1}\}}$ is more than enough to deplete both the remaining supply at $x_t$ as well as $x_{t+1}$. This is due to the fact that although $\Facilities_t$ is under-demanded, $\Facilities_{t+1}$ is not (otherwise \GenAlgorithm\ would have chosen $\Facilities_{t+1}$ instead of $\Facilities_t$ in~\cref{step:maximality}), which implies that $|\agentsOf{\Facilities_{t+1}}| > \capOf{\Facilities_{t+1}}$. In order to determine how to split the agents of $\agentsOf{\{x_{t+1}\}}$ between $x_t$, $ x_{t+1}$, and subsequent facilities, the algorithm simulates $\AOLAlgThreeSF$ with the true remaining capacities for $x_t$ and $x_{t+1}$, and a capacity for $x_{t+1}$ such that the total capacity equals the size $|\agentsOf{\{x_{t+1}\}}|$. This determines which agents will be assigned to $x_t$ and $x_{t+1}$, and that portion of the assignment is finalized; no assignment is made to $x_{t+2}$ yet. 

Finally, once the partial assignment $\matching'$ is determined as described above, \GenAlgorithm\ is called recursively on the remaining unassigned agents and capacities, leading to an assignment $\matching''$ that is combined with $\matching'$, and returned as $\matching$.

\begin{algorithm}[t]
\caption{\GenAlgorithm}
\label{alg:general-locations}
\begin{algorithmic}[1]
\Require Set of facilities $F$ ordered so that their locations satisfy $\ell_{x_1}<\dots<\ell_{x_f}$, capacities $(c_x)_{x\in F}$, and agents $\Agents$ that provide $\ORD$ and $\APP$ information with threshold $\alpha=3$.
\Ensure Matching $\matching$.
\State Find the largest $t$ such that the prefix $\Facilities_t = \{x_1,\dots,x_t\}$ is demand-monotone. \label{step:maximality}
\If{$|\agentsOf{\Facilities_t}| = \capOf{\Facilities_t}$} 
        \State $c'_x \gets c_x$ for each $x\in \Facilities_t$ and $c'_x \gets 0$ otherwise.
        \State $N'\gets \agentsOf{\Facilities_t}$.
        \State $\matching' \gets 3$-\AOLAlgPolarized$(F, c', N')$.
\ElsIf{$|\agentsOf{\Facilities_t}| > \capOf{\Facilities_t}$}
        \State $e\gets  |\agentsOf{\Facilities_t}| - \capOf{\Facilities_t}$ be the excess demand
        \State $c'_x \gets {c_x}$ for each $x\in \Facilities_t$, $c'_{x_{t+1}} \gets e$ for facility $x_{t+1}$, and $c'_x \gets 0$ otherwise.
        \State $N'\gets \agentsOf{\Facilities_t}$.
        \State $\matching' \gets 3$-\AOLAlgPolarized$(F, c', N')$.
\ElsIf{$|\agentsOf{\Facilities_t}| < \capOf{\Facilities_t}$}
        \State $e\gets  \capOf{\Facilities_t} - |\agentsOf{\Facilities_t}|$ be the excess supply (capacity)
        \State $N_1\gets \agentsOf{\Facilities_t}$.
        \State $c'_x \gets c_x$ for each $x\in \Facilities_{s-1}$, $c'_{x_t} \gets c_{x_t}-e$ for facility $x_t$, and $c_x \gets 0$ otherwise.
        \State $\matching' \gets 3$-\AOLAlgPolarized$(F, c', N')$.
        \State $N_2 \gets \agentsOf{\{x_{t+1}\}}$.
        \State $\tilde{c}_{x_t}\gets e$, $\tilde{c}_{x_{t+1}}\gets c_{x_{t+1}}$, $\tilde{c}_{x_{s+2}} \gets |\agentsOf{\{x_{t+1}\}}|- c_{x_{t+1}}-e$, and $\tilde{c}_x\gets 0$ otherwise.
        \State Simulate $\AOLAlgThreeSF((x_t, x_{t+1}, x_{t+2}), \tilde{c}, N_2)$.
        \State $B_t, B_{t+1} \gets$ agents in $N_2$ that \AOLAlgThreeSF\ would match to $x_t$ and $x_{t+1}$, resp. \label{line:3Fcall}
        \State $\IndividualMatching'_i \gets x_t$ for $i\in B_t$ and $\IndividualMatching'_i \gets x_{t+1}$ for $i\in B_{t+1}$.
        \State $c'_{x_t} \gets c'_{x_t} + \tilde{c}_{x_t}$ and  $c'_{x_{t+1}} \gets c'_{x_{t+1}} + \tilde{c}_{x_{t+1}}$.
        \State $N'\gets N_1\cup B_t \cup B_{t+1}$.
\EndIf
\State $N''\gets N\setminus N'$ and $c''_x \gets c_x - c'_x$ for each $x\in \Facilities$.
\State $\matching'' \gets \text{\GenAlgorithm}(F, c'', N'')$.
\State Combine $\matching'$ with $\matching''$ into $\matching$.
\State \Return $\matching$
\end{algorithmic}
\end{algorithm}

\begin{theorem}\label{thm:distances+approvals+ordinal:general-upper}
    The distortion of {\normalfont \GenAlgorithm} is at most $2$ for any instance.
\end{theorem}

\begin{proof}
    We first observe that, just like the \AOLAlgPolarized\ algorithm, \GenAlgorithm\ also iterates over facilities from left to right and assigns agents to these facilities using the ordering $\pi_\Agents$ from \cref{line:ordering} of \AOLAlgPolarized. This is because many of the assignments in \GenAlgorithm\ are actually determined by calling \AOLAlgPolarized\ (which explicitly uses $\pi_\Agents$) on demand-monotone prefixes of the general instance. In fact, the only exception is after \cref{line:3Fcall}, where \GenAlgorithm\ simulates the \AOLAlgThreeSF\  algorithm. However, even this simulation focuses on assigning the agents of $\agentsOf{\{x_{t+1}\}}$ only after all of the agents of $\agentsOf{\Facilities_t}$ have been assigned already, and before any agents from $\agentsOf{\{x_{t'}\}}$ with $t'>t$ have been assigned. As a result, just like we argued in the proof of \cref{thm:polarized-match} for \AOLAlgPolarized, the assignment $\matching$ returned by \GenAlgorithm\ allocates each set $\agentsOf{\{x_t\}}$ to the same set of facilities as the optimal matching $\OptMatching$ from~\cref{lem:line:optimal}, i.e., it satisfies \cref{eq:same_distribution}.
    
    Given this observation, note that, whenever \GenAlgorithm\ calls \AOLAlgPolarized, it provides as input the set of agents $\agentsOf{\Facilities_t}$ along with exactly the same set of facility capacities that this set of agents occupy in $\OptMatching$. The \AOLAlgPolarized\ algorithm then determines a matching of these agents to these facilities, and the (globally) optimal matching is available for the algorithm to choose (since the optimal facilities and corresponding capacities are provided as input). As a result, for every set $\agentsOf{\{x_t\}}$ whose matching is determined by calling \AOLAlgPolarized\  with $\alpha=3$, \cref{thm:polarized-match} implies that
    $$\sum_{i\in \agentsOf{\{x_t}\}} d(i, \IndividualMatching_i) \leq 2 \cdot \sum_{i\in \agentsOf{\{x_t}\}} d(i, \IndividualOpt_i),$$
    since $\max\left\{ 1 + \frac{2}{\alpha}, \frac{3\alpha-1}{\alpha+1} \right\}=2$ for $\alpha=3$.

    What remains to show is that distortion of $2$ is also guaranteed for sets $\agentsOf{\{x_t\}}$ whose matching is determined by simulating the \AOLAlgThreeSF\ algorithm in \cref{line:3Fcall}. 
    Let $B_t$, $B_{t+1}$, and $B_{t+2}$ be the agents from $\agentsOf{\{x_{t+1}\}}$ that the simulation of the \AOLAlgThreeSF\ algorithm in \cref{line:3Fcall} would assign to facilities $x_t$, $x_{t+1}$, and  $x_{t+2}$, respectively. Note that $B_t$, $B_{t+1}$, and $B_{t+2}$ is a partition of $\agentsOf{\{x_{t+1}\}}$. Let $O_t:=\{i\in \agentsOf{\{x_{t+1}\}}~|~ o_i = x_t\}$ and $O_{t+1}:=\{i\in \agentsOf{\{x_{t+1}\}}~|~ o_i = x_{t+1}\}$ be the sets of agents from $\agentsOf{\{x_{t+1}\}}$ that $\OptMatching$ assigns to $x_t$ and $x_{t+1}$, respectively. Also, let $O_{t+2}=\agentsOf{\{x_{t+1}\}}\setminus (O_t\cup O_{t+1})$ be all the other agents in $\agentsOf{\{x_{t+1}\}}$ that $\OptMatching$ assigns to facilities weakly to the right of $x_{t+2}$. Note that since $\matching$ and $\OptMatching$ assign the same number of agents from $\agentsOf{\{x_{t+1}\}}$ to each facility (\cref{eq:same_distribution}), then $|O_t|=|B_t|$, $|O_{t+1}|=|B_{t+1}|$ (and thus also $|O_{t+2}|=|B_{t+2}|$). By \cref{thm:3Fdistortion}, we get that the cost of the assignment recommended by the simulation of \AOLAlgThreeSF\ is no more than twice the cost of any other assignment, i.e., 
    \begin{align}
        &\sum_{i\in B_t} d(i, x_t) + \sum_{i\in B_{t+1}} d(i, x_{t+1}) +\sum_{i\in B_{t+2}} d(i, x_{t+2}) \nonumber \\
        &\leq 2\cdot \left( \sum_{i\in O_t} d(i, x_t) + \sum_{i\in O_{t+1}} d(i, x_{t+1})+\sum_{i\in O_{t+2}} d(i, x_{t+2})\right). \label{eq:approx_guarantee2}
    \end{align}
    Since $\matching$ does assign $B_t$ to $x_t$ and $B_{t+1}$ to $x_{t+1}$, but assigns $B_{t+2}$ to facilities weakly (and maybe strictly) to the right of $x_{t+2}$ we can conclude that the total cost of $\matching$ is
    \begin{align*}
    \sum_{i\in \agentsOf{\{x_{t+1}\}}} d(i,\IndividualMatching_i) &= \sum_{i\in B_t} d(i, x_t) + \sum_{i\in B_{t+1}} d(i, x_{t+1}) +\sum_{i\in B_{t+2}} d(i, \IndividualMatching_i)\\
    &=\sum_{i\in B_t} d(i, x_t) + \sum_{i\in B_{t+1}} d(i, x_{t+1})+ \sum_{i\in B_{t+2}} d(i, x_{t+2}) + \sum_{i\in B_{t+2}} d(x_{t+2},\IndividualMatching_i) \\
    &\leq 2\cdot \left( \sum_{i\in O_t} d(i, x_t) + \sum_{i\in O_{t+1}} d(i, x_{t+1})+\sum_{i\in O_{t+2}} d(i, x_{t+2})\right)+ \sum_{i\in B_{t+2}} d(x_{t+2},\IndividualMatching_i)\\
    &= 2\cdot \left( \sum_{i\in O_t} d(i, x_t) + \sum_{i\in O_{t+1}} d(i, x_{t+1})+\sum_{i\in O_{t+2}} d(i, x_{t+2})\right)+ \sum_{i\in O_{t+2}} d(x_{t+2},\IndividualOpt_i)\\
    &= 2\cdot \sum_{i\in \agentsOf{\{x_{t+1}\}}} d(i,\IndividualOpt_i),
    \end{align*}
    where the second equation uses the fact that the agents in $B_{t+2}$ are to the left of $x_{t+2}$ but are matched to facilities $\IndividualMatching_i$ weakly to the right of $x_{t+2}$, the inequality is due to \cref{eq:approx_guarantee2}, and the subsequent equation uses the fact that, from \cref{eq:same_distribution}, both $\matching$ and $\OptMatching$ assign  agents to the same facilities weakly to the right of $x_{t+1}$, so $\sum_{i\in B_{t+2}} d(x_{t+2},\IndividualMatching_i)= \sum_{i\in O_{t+2}} d(x_{t+2},\IndividualOpt_i)$.   
\end{proof}

\subsection{Tight Lower Bound}\label{sec:distances+approvals+ordinal:lower}
We now present a lower bound of $2$ for all algorithms that use all three types of information, showing that our algorithm is best possible within this class.  

\begin{theorem}\label{thm:distances+approvals+ordinal:three-locations:lower}
The distortion of any algorithm that uses $\DIST + \APP + \ORD$ information is at least $2$, even when there are $f=3$ facilities. 
\end{theorem}

\begin{proof}
If $\alpha \geq 3$, then, by \cref{lem:distances+approvals+ordinal:lower:large-alpha}, the distortion is at least 
\begin{align*}
    \frac{3\alpha-1}{\alpha+1} \geq 2.
\end{align*}
So, it suffices to consider the case $\alpha < 3$.
Let $\varepsilon > 0$ be an infinitesimal and consider the following instance: 
\begin{itemize}
    \item There is a facility $L$ with capacity $c_L = (n-1)/2$ at $-\alpha-\varepsilon$;
    \item There is a unit-capacity facility $x$ at $1$;
    \item There is a facility $R$ with capacity $c_R = (n-1)/2$ at $\alpha$.
\end{itemize}
All $n$ agents have the same ordinal preference such that $x$ is their most-preferred item, followed by $R$, and then $L$; that is, $x \succ_i R \succ_i L$ for any $i$. In addition, all agents approve $x$ and $R$, that is, $A_i=\{x, R\}$ for any $i$. 
Consider any matching $\matching$ that might be computed by an arbitrary algorithm. We place the agents on the line as follows (see \cref{fig:distances+approvals+ordinal:three-locations:lower}).

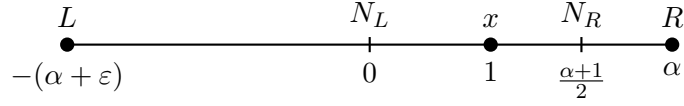
\begin{figure}[t]
    \centering
    \begin{tikzpicture}[scale=1]
        \draw[thick] (-4,0) -- (4,0);

        \draw[thick,fill=black] (-4,0) circle (0.08);
        \node[below] at (-4, -0.1) {$-(\alpha+\varepsilon)$};
        \node[above] at (-4, 0.1) {$L$};

        \draw[thick] (0,0.1) -- (0, -0.1) node[below] {$0$};
        \node[above] at (0, 0.1) {$N_L$};

        \draw[thick,fill=black] (1.6,0) circle (0.08);
        \node[below] at (1.6, -0.1) {$1$};
        \node[above] at (1.6, 0.1) {$x$};

        \draw[thick] (2.8,0.1) -- (2.8, -0.1) node[below] {$\frac{\alpha+1}{2}$};
        \node[above] at (2.8, 0.1) {$N_R$};

        \draw[thick,fill=black] (4,0) circle (0.08);
        \node[below] at (4, -0.1) {$\alpha$};
        \node[above] at (4, 0.1) {$R$};
    \end{tikzpicture}
\caption{The metric space considered in the proof of \cref{thm:distances+approvals+ordinal:three-locations:lower}.}
    \label{fig:distances+approvals+ordinal:three-locations:lower}
\end{figure}

\begin{itemize}
    \item The agents that are matched to $R$ are located at $0$; let $N_L$ be the set of these agents. 
    Observe that for any $i\in N_L$, we have that
    \begin{align*}
    & d(i,L) = \alpha+\varepsilon; \\
    & d(i,x) = 1; \\
    & d(i,R) = \alpha.
    \end{align*}
    Hence, this position of $i$ is consistent to the agent's ordinal and approval preferences.
    
    \item The agents that are matched to $x$ and $L$ are located at $\frac{\alpha+1}{2}$; let $N_R$ be the set of these agents. 
    For any $i \in N_R$, we have that
    \begin{align*}
    & d(i,L) = \frac{3\alpha+1}{2}+\varepsilon; \\
    & d(i,x)= \frac{\alpha-1}{2}; \\
    & d(i,R)= \frac{\alpha-1}{2}.
    \end{align*}
    Hence, this position of $i$ is consistent to the agent's ordinal and approval preferences. 
\end{itemize}
The social cost of the computed matching is 
\begin{align*}
    \SC(\matching) 
    &= \frac{n-1}{2} \cdot \alpha + \frac{\alpha-1}{2} + \frac{n-1}{2} \cdot \left( \frac{3\alpha+1}{2}+ \varepsilon \right) \\
    &\geq \frac{n-1}{2} \cdot \frac{5\alpha+1}{2}.
\end{align*}
On the other hand, the optimal matching $\OptMatching$ is to match the $(n-1)/2$ the agents in $N_L$ to $L$, one of the agents in $N_R$ to $x$, and the remaining $(n-1)/2$ agents in $N_R$ to $R$. The optimal social cost is 
\begin{align*}
    \SC(\OptMatching) 
    &= \frac{n-1}{2} \cdot (\alpha+\varepsilon) + \left( \frac{n-1}{2}+1 \right) \cdot \frac{\alpha-1}{2} \\
    &\leq \frac{n+1}{2} \cdot \frac{3\alpha-1}{2}. 
\end{align*}
As $n$ tends to infinity, and since $\alpha <3$, the distortion becomes 
\begin{align*}
    \frac{5\alpha+1}{3\alpha-1} > 2,
\end{align*}
and the proof is complete.
\end{proof}

\section{A Note on the Necessity of $\DIST$ Information}
\label{sec:approvals+ordinal:>=4facilities:lower:3}
As we saw in the previous sections, for instances with an arbitrary number of facilities, it is possible to improve upon the best possible bound of $3$ that can be achieved by purely ordinal algorithms using different or additional information. In \cref{sec:distances+approvals}, we showed that a distortion of $1+\sqrt{2}$ can be achieved with $\DIST+\APP$ information, and combining these two types of information is necessary as, on their own, they do not lead to any meaningful distortion bound. In \cref{sec:>=3-facilities:distances+approvals+ordinal}, we showed that an even further improved distortion bound of $2$ can be achieved using $\DIST+\APP+\ORD$ information. Naturally, one might wonder what distortion bounds are possible using $\APP+\ORD$ information, a combination that we have not discussed much, with the exception of instances with two facilities in \cref{sec:distances+approvals+ordinal:two}. Unfortunately, as we show here, when there are at least four facilities, it is impossible to achieve a distortion better than $3$ using only $\APP+\ORD$ information. This implies that the additional $\APP$ information is not useful, and $\DIST$ information is also necessary to obtain an improvement; in fact, as additional information, we only need the ordering of the facilities (rather than their exact pairwise distances), and this ordering cannot be computed using only $\APP+\ORD$. 

\begin{theorem}\label{thm:approvals+ordinal:four:lower}
    The distortion of any matching algorithm that uses $\ORD+\APP$ information is at least $3$, even when there are $\NumberFacilities=4$ facilities.
\end{theorem}

\begin{proof}
    Consider an instance with four unit-capacity facilities $\{x, y, z, w\}$, and four agents. Two of them have ranking $x\succ y \succ z \succ w$, while the other two have $z \succ y \succ x \succ w$. Every agent reports only their top-ranked facility in their approval set, irrespective of the choice of $\alpha$.

    For any algorithm that assigns an agent from the second pair to facility $y$, consider a metric where the facilities $x, y, z$, and $w$ are located at $0, \varepsilon, 1$, and $2+\varepsilon$, respectively, for some arbitrarily small $\varepsilon>0$. If the first pair of agents is located at $0$ and the second pair is located at $1$, the $\ORD+\APP$ information they would provide is consistent with what we described above. The optimal assignment would match the first two agents to $x$ and $y$ and the other two agents to $z$ and $w$, leading to a social cost of $1+2\varepsilon$. However, any algorithm that matches an agent from the first pair to $y$ is forced to suffer a social cost of at least $3$ (a cost of $1-\varepsilon$ to connect to $y$, and a total cost of $2+\varepsilon$ to connect to $z$ and $w$). This leads to a distortion of at least $3/(1+2\varepsilon)$.

    Similarly, for any algorithm that assigns an agent from the first pair to facility $y$, consider a metric where the facilities $x, y, z$, and $w$ are located at $0, 1-\varepsilon, 1$, and $-1$, respectively. If the agents are at the same locations as before, the $\ORD+\APP$ information they would provide would again be consistent with what was described above. This time, the optimal assignment would match the first two agents to $x$ and $z$, and the other two agents to $y$ and $w$, leading to a social cost of $1+2\varepsilon$. However, any algorithm that matches an agent from the first pair to $y$ is forced to suffer a social cost of $3$ (a cost of $1-\varepsilon$ to connect to $y$, and a total cost of $2+\varepsilon$ to connect to $x$ and $w$). This once again leads to a distortion of at least $3/(1+2\varepsilon)$. Since $\varepsilon$ can be arbitrarily small, this concludes the proof.
\end{proof}

Interestingly, the above lower bound requires instances with four facilities, and leaves open the possibility that, using only $\APP+\ORD$ information, a distortion bound smaller than $3$ might be achievable for instances with three facilities. Showing this is, however, a quite challenging task and a question that our work leaves open. 

\section{Conclusion and Future Directions}
We studied a fundamental metric facility assignment problem and explored how the performance of deterministic algorithms depends on the partial information available regarding the underlying metric. We established tight distortion bounds for algorithms using different combinations of ordinal preferences, approval preferences, and inter-facility distances. With few exceptions, our results apply to the line metric. A natural open question is therefore to understand the power of these additional information types beyond the line, and in particular whether they can be used to improve the best known bounds for deterministic or randomized metric matching in general metrics. Another promising direction is adaptive information elicitation, such as allowing the algorithm to choose the approval parameter $\alpha$ after observing the ordinal preferences and/or the inter-facility distances. Finally, our framework could be extended to other social objectives. While we focused on the standard social cost, it would be interesting to study the same informational models for the $k$-centrum objective and other fairness-inspired objectives, as has been done for ordinal algorithms~\citep{Filos-RatsikasG25,hastings2025fairmetricdistortionmatching}.

\bibliographystyle{plainnat}
\bibliography{references}

\end{document}